%% file: main.tex
\documentclass{article}[11pt]
\usepackage[sort]{cite}
\usepackage{geometry}
\input{preamble}

\input{bib_macros.tex}

\title{Efficiently Constructing Sparse Navigable Graphs}

\author{
    Alex Conway\thanks{Cornell Tech (\url{ajc473@cornell.edu}, \url{btl46@cornell.edu})} \and
    Laxman Dhulipala\thanks{University of Maryland (\url{laxman@umd.edu}, \url{rwen1@umd.edu})}\and
    Martin Farach-Colton\thanks{New York University (\url{mlf9579@nyu.edu}, \url{cmusco@nyu.edu}, \url{ys6378@nyu.edu}, \url{torsten.suel@nyu.edu})}\and
    Rob Johnson\thanks{VMware Research (\url{rob.johnson@broadcom.com})} \and
    Ben Landrum\footnotemark[1]\and
    Christopher Musco\footnotemark[3]\and
    Yarin Shechter\footnotemark[3]\and
    Torsten Suel\footnotemark[3]\and
    Richard Wen\footnotemark[2] 
}

\date{}

\begin{document}

\maketitle

\begin{abstract}
    Graph-based nearest neighbor search methods have seen a surge of popularity in recent years, offering state-of-the-art performance across a wide variety of applications. Central to these methods is the task of constructing a \emph{sparse navigable search graph} for a given dataset endowed with a distance function. Unfortunately, doing so is computationally expensive, so heuristics are universally used in practice. 

    In this work, we initiate the study of fast algorithms with provable guarantees for search graph construction. For a dataset with $n$ data points, the problem of constructing an optimally sparse navigable graph can be framed as $n$ separate but highly correlated minimum set cover instances. This yields a naive $O(n^3)$ time greedy algorithm that returns a navigable graph whose sparsity is at most $O(\log n)$ higher than optimal. We improve significantly on this baseline, taking advantage of correlation between the set cover instances to leverage techniques from streaming and sublinear-time set cover algorithms. By also introducing problem-specific pre-processing techniques, we obtain an $\tilde{O}(n^2)$ time algorithm for constructing an $O(\log n)$-approximate sparsest navigable graph under any distance function. 

    The runtime of our method is optimal up to logarithmic factors under the Strong Exponential Time Hypothesis via a reduction from the Monochromatic Closest Pair problem. Moreover, we prove that, as with general set cover, obtaining better than an $O(\log n)$-approximation is NP-hard, despite the significant additional structure present in the navigable graph problem. Finally, we show that our approach can also beat cubic time for the closely related and practically important problems of constructing $\alpha$-shortcut reachable and $\tau$-monotonic graphs, which are also used for nearest neighbor search. For such graphs, we show that black-box sublinear set cover algorithms yield $\tilde{O}(n^{2.5})$ time or better algorithms. 
\end{abstract}

\thispagestyle{empty}
\newpage
\setcounter{page}{1}

\section{Introduction}
Nearest neighbor search  has been central in computer science and data science for decades~\cite{IndykMotwani:1998,Kleinberg:1997,KushilevitzOstrovskyRabani:1998,Charikar:2002,BeygelzimerKakadeLangford:2006,AndoniIndyk:2008,JegouDouzeSchmid:2011,AndoniIndykNguyen:2014,AumullerBernhardssonFaithfull:2020}. Current interest in the problem is driven by its central role in modern AI systems: search over machine-learned semantic embeddings is the key technology behind AI-driven web search, retrieval-augmented generation in LLMs, and more \cite{LewisPerezPiktus:2020,KitaevKaiserLevskaya:2020,XiongXiongLi:2021,LuoLakshmanShrivastava:2022,Bruch:2024}. 

Formally, the nearest neighbor search problem is defined as follows: Given a point set $P = \{p_1, \ldots, p_n\}$ and a distance function $d(\cdot,\cdot)$,\footnote{Our work only requires that $d(p_i,p_i) = 0$ for all $i$ and $d(p_i,p_j) = d(p_j,p_i) >0$ for $i \neq j$ (i.e., we assume no repeated data points in the dataset). These conditions are satisfied by any metric, but are much weaker.}  construct a data structure to efficiently find $\argmin_{i\in \{1, \ldots, n\}} d(p_i, q)$ for any given query point $q$, or an approximate minimizer. A naive $O(n)$ time method is to brute-force search over the entire dataset, so ``efficiently'' here means finding the nearest neighbor in $o(n)$ time.

The nearest neighbor search problem has been studied extensively both in theory and in practice.
A classical solution is locality-sensitive hashing (LSH)~\cite{IndykMotwani:1998,AndoniIndykLaarhoven:2015}, which provides theoretical guarantees on the approximation ratio, construction time, and query time.
However, LSH has proven too inefficient for large high-dimensional datasets~\cite{ManoharShenBlelloch:2024}. For these applications \defn{graph-based} nearest neighbor search heuristics like the Hierarchical Navigable Small World (HNSW) 
method \cite{MalkovYashunin:2020}, Vamana/DiskANN \cite{SubramanyaDevvritKadekodi:2019,KrishnaswamyManoharSimhadri:2024}, 
Navigating Spreading-out Graphs (NSG) \cite{FuXiangWang:2019}, and more \cite{MalkovPonomarenkoLogvinov:2014,FuCai:2016,HarwoodDrummond:2016,PengChoiChan:2023,ManoharShenBlelloch:2024,AumullerBernhardssonFaithfull:2020,SimhadriAumullerIngber:2024} have become much more popular.
While these graph-based heuristics substantially outperform other (potentially more theoretically sound) techniques in practice, there is surprisingly little underlying theory about their performance.

Although the details differ, the core idea behind all of these graph-based search methods is the same, dating back to decades old work on Milgram's ``small world phenomenon'' and greedy routing in networks \cite{Milgram:1967,WattsStrogatz:1998,Kleinberg:2000,Kleinberg:2000b}. The basic approach is to preprocess the input dataset by constructing a \emph{directed search graph} $G=\left(P,E\right)$, with a node corresponding to each point  $p_i$. To answer a query $q$, {\bf \emph{greedy search}}, or some variant like\footnote{Beam-search, a natural ``conservative'' version of greedy search is currently the method of choice in practice.}, is employed: traverse a path through the graph, starting at an arbitrary node, until a local minimum condition is met.  To extend a path at node $p$, move to $p$'s neighbor closest to $q$;  the local minimum is achieved (and the path is terminated) when the current $p$ has no neighbors that are closer to $q$ than $p$ is.  We then return $p$. The main computation cost at each step is finding $p$'s closest neighbor to $q$, which requires a number of distance computations proportional to $p$'s degree.

A key question in graph-based search is how to choose $G$, subject to the following competing constraints: 
\begin{itemize}
%[label=\arabic*.,topsep=0pt,itemsep=0pt,parsep=0pt,leftmargin=15pt]
%[nolistsep,noitemsep]
    \item Can we construct a graph $G$ whose greedy traversal will yield a nearest neighbor, or good approximate nearest neighbor, of the query point?

    \item Can we construct $G$ to be \emph{sparse} so that memory overhead is minimized and traversing it is fast?

\end{itemize}
We note that these two constraints are indeed in tension: if $G$ were a clique, it would certainly always return the nearest neighbor of a query, but traversing it would take $O(n)$ time for every query.  If, on the other hand, $G$ were a constant-degree spanner, greedily traversing $G$ would require just $O(1)$ time per step and the graph would require just $O(n)$ space, but might not return a very good answer to most queries.

Our work studies the fundamental question of how to efficiently construct graphs that satisfy the competing constraints.  We explain the notions required to make this statement precise in the next section.

\subsection{Navigable Search Graphs}
\label{sec:navigable_graphs}
As discussed above, our goal is to construct a search graph $G$ where greedy search is effective in finding (approximate) nearest neighbors in a point set $P$. Towards this goal, many different graphs properties have been considered (see \cite{PengChoiChan:2023} for an overview). The most popular such property is \emph{navigability}. Since convergence to a nearest-neighbor for all possible queries is difficult to guarantee\footnote{Despite their practical effectiveness, most graph-based methods do not enjoy worst-case theoretical accuracy guarantees as strong as, for example, methods based on LSH \cite{AndoniIndyk:2008}.}, navigability demands that, at least if the query $q$ is \emph{in $P$} (i.e., it has a nearest neighbor with distance $0$), then greedy search on $G$ should return $q$ itself, no matter where the search begins. A standard equivalent definition for navigability is as follows \cite{Al-JazzaziDiwanGou:2025}: 
\begin{definition}[Navigable Graph]
    \label{def:navigability}
Given point set $P$ and distance function $d : P\times P \to \mathbb{R}^+$, we say that a directed graph $G=\left(P,E\right)$ is \emph{navigable} if, for all $p_i,p_k\in P$, $p_i\neq p_k$, there exists $p_j\in P$ such that $(p_i,p_j)\in E$ and $d(p_j,p_k) < d(p_i,p_k)$. That is, $p_i$ has an out-neighbor, $p_j$, that is closer to $p_k$ than $p_i$ is.
\end{definition}

Navigability has been studied in the context of nearest neighbor search for decades \cite{JaromczykToussaint:1992,AryaMount:1993} and lends its name to modern methods such as HNSW (Hierarchical Navigable Small World) and NSG (Navigating Spreading-out Graphs). 
Moreover, several other well-studied graph properties are direct generalizations of navigability. For example, the popular DiskANN method \cite{KrishnaswamyManoharSimhadri:2024} aims to construct ``$\alpha$-shortcut reachable graphs'' \cite{IndykXu:2023}, which strengthen \Cref{def:navigability} to require that $d(p_k,p_j) < \frac{1}{\alpha}d(p_i,p_j)$ for some $\alpha \geq 1$. Similarly, the $\tau$-monotonicity property of \cite{PengChoiChan:2023} requires $d(p_k,p_j) < d(p_i,p_j) - \tau$ for some $\tau \geq 0$. These strengthenings of navigability are interesting because, while they require a denser graph to satisfy, they can lead to strong approximation guarantees. For example, Indyk and Xu recently showed that greedy search in an $\alpha$-shortcut reachable graph always returns an approximate nearest neighbor $\tilde{p}$ with  $d(q,\tilde{p}) \leq \frac{\alpha+1}{\alpha-1} \cdot \min_{p\in P} d(q,p)$ \cite{IndykXu:2023}. Greedy search in $\tau$-monotonic graphs always returns an exact nearest neighbor if $\min_{p\in P} d(q,p) \leq \tau/2$ \cite{PengChoiChan:2023}. Recent work also shows that multiplicative approximation guarantees for general metrics can be obtained when a variant greedy search called Adaptive Beam Search is run on a vanilla navigable graph \cite{Al-JazzaziDiwanGou:2025}

Graph navigability is also studied as a natural concept beyond applications to vector search. 
Milgram's famous ``small-world'' experiments from the 1960s and 70s established that social networks are empirically navigable (under an implicit distance function) \cite{Milgram:1967,TraversMilgram:1977}. 
This finding prompted follow-up work on  understanding the navigability of other real-world networks \cite{DoddsMuhamadWatts:2003,BogunaKrioukovClaffy:2009}, 
as well as theoretical work by Kleinberg and others on generative models for graphs and graph evolution that lead to navigability \cite{Kleinberg:2000,Kleinberg:2000b, ClausetMoore:2003,ChaintreauFraigniaudLebhar:2008,HuWangLi:2011,GulyasBiroKorosi:2015}.

\subsection{Efficient Graph Construction}
Despite the interest in navigability and its practical importance, surprisingly no prior work has addressed one of the most basic algorithmic questions to ask about the property:
\begin{center}
    \emph{
How efficiently can one compute the \ul{sparsest navigable graph} for a point set $P$ and distance function $d$?}
\end{center}
Note that ``sparsest'' can be defined in multiple ways: we could ask to minimize maximum degree or average degree. As will become clear in \Cref{sec:nav_graph_set_cover}, it turns out that a single graph can obtain minimum degree simultaneously for all $i$, so both these optimization criteria can be achieved simultaneously.

Current graph-based nearest neighbor search methods universally employ heuristics that tend to construct ``almost'' navigable graphs \cite{DiwanGouMusco:2024}. Some of these heuristics are based on natural relaxations of methods that construct provably navigable graphs, such as the ${O}(n^3)$ time  RobustPrune method (with complete candidate set) from DiskANN \cite{KrishnaswamyManoharSimhadri:2024}, and related methods proposed in other  papers \cite{HarwoodDrummond:2016,FuXiangWang:2019,PengChoiChan:2023}. However, not only are these methods slow (cubic time), but even before relaxation, there is no guarantee that they construct the sparsest navigable graph, or a graph whose sparsity is competitive with the optimum. Indeed, recent work provides examples showing that DiskANN can return a solution $\tilde{\Omega}(n)$ times worse than optimal \cite{KhannaPadakiWaingarten:2025}.

A recent paper by Diwan et al. shows that \emph{any} set with $n$ unique points admits a navigable graph with $\tilde{O}(n^{1.5})$ edges, and they show how to construct such a graph in $\tilde{O}(n^2)$ time \cite{DiwanGouMusco:2024}. However, while it was shown that ${O}(n^{1.5})$ is optimal in the worst-case, as observed in follow-up work \cite{Al-JazzaziDiwanGou:2025}, typical datasets admit far sparser navigable graphs. The algorithm from \cite{DiwanGouMusco:2024} is not able to take advantage of the fact that a sparser graph might exist for a given $P$. Similarly, Arya and Mount give an efficient algorithm that for any point set in $d$-dimensional Euclidean space, constructs a navigable graph with maximum degree $2^{O(d)}$ \cite{AryaMount:1993}. While this is tight in the worst-case, for any particular point set, we could hope to do much better.

\subsection{Our Results}
In this paper, we make progress on the stronger goal stated above of constructing graphs with \defn{instance-optimal} sparsity.
As will be discussed in \Cref{sec:tech_overview}, the problem of constructing a minimum sparsity navigable graph can be framed as $n$ different minimum set cover instances, each involving $n$ elements and $n$ sets. Running the standard greedy approximation algorithm for set cover on each instance thus yields an $O(\log n)$ approximation in $O(n^3)$ time.  Our main contribution is to push significantly beyond this bound. In particular, we provide a nearly-quadratic time algorithm:
\begin{theorem}
\label{thm:main_navigability}
There is a randomized algorithm running in $\tilde{O}(n^2)$ time\footnote{We assume $d(\cdot,\cdot)$ can be evaluated in $O(1)$ time. Formally, if $d(\cdot,\cdot)$ takes $T$ time to evaluate, our method runs in $\tilde{O}(n^2) + O(n^2\cdot T)$ time. Throughout, we use $\tilde{O}(m)$ to denote $O(m\log^c m)$ for a fixed constant $c$.} with high probability that, given any distance function $d$ and point set $P$ in general position\footnote{Formally, we assume $d(p_i, p_j)\neq d(p_i, p_k)$ for any $i,j,k$. In practice, this can be achieved via an arbitrarily small random perturbation of the points. The assumption is necessary to avoid subtleties that arise in the definition of navigability if there can be exactly tied distances. See \Cref{sec:prelims} for details.}, constructs a navigable graph $G$ in which, with high probability, every vertex has a degree which does not exceed that of the sparsest navigable graph for $P$ by more than a multiplicative $c\log n$ factor for a fixed constant $c$.
\end{theorem}
\Cref{thm:main_navigability} is proven in \Cref{sec:main_algo}. The result is surprising in that we do not even know how to check if a given graph is navigable in $\tilde{O}(n^2)$ time. Moreover, under natural complexity-theoretical assumptions, the result is optimal up to logarithmic factors, even when $P$ is in low-dimensional Euclidean space. In particular, any navigable graph $G$ for $P$ \emph{must} contain a bidirectional edge between the two closest points in $P$. It has been established that identifying these two points requires $\Omega(n^{2-\epsilon})$ time under the Strong Exponential Time Hypothesis (SETH) \cite{KarthikManurangsi:2020}. Accordingly, we have the following:
\begin{claim}
    \label{clm:fine_grained_hardness}
    Assuming the Strong Exponential Time Hypothesis (SETH), any algorithm  that, for a set of $n$ points, $P$, returns a navigable graph with $o(n^{2-\epsilon})$ edges requires $\Omega(n^{2-\epsilon})$ time, for any constant $\epsilon > 0$. This claim holds even under the Euclidean distance when $P\subset \R^{O(\log n)}$.
\end{claim}
Recall that all point sets admit a navigable graph with $\tilde{O}(n^{1.5})$ edges \cite{DiwanGouMusco:2024}. \Cref{clm:fine_grained_hardness} shows that beating quadratic time is unlikely even if we just want to slightly improve on the naive solution of returning a complete graph with $O(n^2)$ edges. Obtaining a result with near-minimal sparsity is only harder. 

Moreover, beating the $O(\log n)$ approximation factor in \Cref{thm:main_navigability} is unlikely. Even when $d$ is a metric, which adds additional structure to the set cover instances underlying the minimum navigable graph problem, we are able to leverage the well-known hardness of approximating set cover \cite{DinurSteurer:2014} to prove in  \Cref{sec:lower_bounds}: 
\begin{theorem}
    \label{thm:hardness_of_approx}
   For every $\varepsilon >0$ it is NP-hard to construct a navigable graph $G$ for a given point set $P$ with average out-degree $\left(\frac{1}{3}\ln(n)-\varepsilon\right)\cdot \OPT$, where $\OPT$ is the average out-degree in the minimum navigable graph for $P$. This result holds even under Euclidean distance.
\end{theorem}
% We can obtain a slightly weaker lower bound even for point sets in $\log^c{n}$ dimensional Euclidean space. See Section 6 for more details. 
% \Chris{add discussion about how this is for an easier problem of getting average degree correct.}

Beyond our main algorithmic result (\Cref{thm:main_navigability}), we also observe that the related problems of constructing near optimally sparse $\alpha$-shortcut reachable graphs and $\tau$-monotonic graphs can similarly be framed as $n$ coupled set cover instances. This perspective can be used to obtain faster algorithms for those problems that achieve and $O(\log n)$ approximation factor and run in $\tilde{O}(\min(n^{2.5}), n\cdot \OPT)$, where $\OPT \geq n$ is the sparsity of the optimal solution. Obtaining nearly-quadratic time remains an open question.

\paragraph{Concurrent Work.} We note that Khanna, Padaki, and Waingarten independently obtain a $\tilde{O}(n\cdot \OPT)$ time algorithm for constructing a near optimally sparse $\alpha$-shortcut reachable graphs \cite{KhannaPadakiWaingarten:2025}. Like our work, their approach is based on connecting the problem to set cover. They also define and describe efficient algorithms for an interesting bi-criteria variant of the problem. 

% Beyond the above results on constructing navigable graphs, we also demonstrate how the algorithmic techniques used can be used to match some previously known results for approximately solving the minimum set cover problem in sublinear time. Although our results match the time bounds presented in the existing line of  work on sublinear set cover, \cite{GrigoriadisKhachiyan:1995,NguyenOnak:2008,YoshidaYamamotoIto:2012,KoufogiannakisYoung:2014,IndykMahabadiRubinfeld:2018}, they may be of independent interest as the results are achieved using alternative techniques to those employed in the literature.
% Further, as discussed \Cref{sec:technical_overview}, an observation that is key to our method is that the problem of constructing navigable graphs can be cast as a collection of set cover instances. We show that this is also true for other related graph problems -constructing near -optimally sparse $\alpha$-shortcut reachable graphs and $\tau$-monotonic graphs, which gives rise to sub-cubic algorithms for constructing such graphs that are near-optimally sparse. For details, see \cref{sec:sublinear-set-cover}.

\section{Technical Overview}
\label{sec:tech_overview}

\begin{figure}[ht]
    \centering
    \vspace{-.5em}
\input{set_cover_illustration.tex}
\vspace{-.5em}
    \caption{Illustration of the navigability set cover problem corresponding to node $p_1$. Each image shows the set, $S_{1\rightarrow j}$, for a different choice of out-neighbor $j$. $S_{1\rightarrow j}$ contains all $p_k$ for which $d(p_j,p_k) < d(p_i,p_k)$. Constructing a navigable graph which has the fewest out-edges from node $1$ is equivalent to solving the minimum set cover problem for this instance. Constructing a navigable graph with the fewest total number of edges is equivalent to solving the minimum set cover problem for all $n$ different problem instances.
    }
    \vspace{-.75em}
        \label{fig:set_cover_illustration}
\end{figure}

\subsection{Sparsest Navigable Graph as Minimum Set Cover}
\label{sec:nav_graph_set_cover}
As a starting point towards understanding better algorithms, we first observe that constructing a navigable graph amounts to solving $n$ separate \emph{minimum set cover instances}, one for each node. For the instance corresponding to node $p_i$, the elements that need to be covered are all other points $p_k$, $k \neq i$. 
We also have one set, $S_{i\rightarrow j}$, for each $p_j$, $j\neq i$, which corresponds to adding an edge from $i$ to $j$ in $G$. $S_{i\rightarrow j}$ contains any point $p_k$ such that $d(p_j,p_k) < d(p_i,p_k)$. I.e., to cover an element $p_k$, we need to have an edge from $p_i$ either directly to $p_k$, or to some node that gets us closer to $p_k$. Choosing the smallest number of sets that covers all elements amounts to choosing the smallest number of out-edges for $p_i$ that satisfy the navigability requirements of \Cref{def:navigability}. See \Cref{fig:set_cover_illustration} for an illustration of the set cover instances.

The set cover formulation reinforces that, while navigability is usually thought of as a ``global property'' about routing in $G$, it can be viewed as a collection of purely ``local'' conditions: the out-edges of node $p_i$ can be chosen optimally without considering out-edges for any other nodes in the graph. Moreover, while we are not aware of prior work that leverages the connection to set cover (beyond the concurrent work of \cite{KhannaPadakiWaingarten:2025})
it immediately leads to an approach for constructing a navigable graph with near minimal sparsity: run any black-box approximation algorithm for each of the set cover instances. The best known approximation ratio of $O(\log n)$ can be achieved by the standard greedy algorithm in time linear in the total size of the sets \cite{CormenLeisersonRivest:2022}. 
Implementing this idea naively amounts to $O(n^2)$ time per instance in our setting, for a total complexity of ${O}(n^3)$ time to construct a navigable graph where the degree of every node is at most $O(\log n)$ higher than optimal. By simultaneously obtaining near-minimum degree for every node, it follows that the greedy algorithm produces a graph $G$ whose maximum and average degree are within $O(\log n)$ of optimal.

\subsection{Algorithmic Techniques}
\label{sec:technical_overview}
Our algorithm that establishes \Cref{thm:main_navigability} relies on three key ingredients to push beyond the $O(n^3)$ baseline obtained via greedy set cover. These ingredients are summarized below.

\medskip \noindent \textbf{Efficient Query Access.} We first observe that, instead of treating navigable graph construction as $n$ disjoint set cover instances, we can obtain a faster runtime by taking advantage of the fact that these instances are highly correlated. 
Concretely, we can \emph{jointly preprocess} the instances in $O(n^2\log n)$ time to support natural ``query access'' to the set cover problem, which obviates the need to explicitly construct each set cover instance (write down each set and its elements), as doing so would already require $\Omega(n^3)$ time. 

Specifically, for each point, we construct a list of all other points, sorted in increasing order of distance from that point. Doing so takes $O(n\log n)$ time per data point, so $O(n^2\log n)$ time total. 
From these lists, we can immediately answer what is called a ``SetOf'' query in the literature on sublinear-time set cover algorithms \cite{NguyenOnak:2008,YoshidaYamamotoIto:2012,IndykMahabadiRubinfeld:2018}: concretely, for any element $p_k$, we can return in $O(1)$ time the number of sets, $q$, that cover the element, and the identity of the $\ell^{\text{th}}$ set covering $p_k$ for any $\ell \leq q$. 
To do so, we simply look at the sorted list corresponding to the point $p_k$.
For the set cover instance corresponding to node $i$, the number of points before $p_i$ in $p_k$'s list is $q$, as this is the number of points $j$ for which $d(p_j,p_k) < d(p_i,p_k)$. 
Moreover, the $\ell^{\text{th}}$ set that covers $p_k$ is simply the set corresponding to the $\ell^\text{th}$ element of the list.
In addition to SetOf queries, presorting our data also allows us to easily answer ``Membership'' queries in $O(1)$ time: to check if $p_k\in S_{i\rightarrow j}$, simply check if $p_j$ precedes $p_i$ in the sorted list corresponding to $p_k$. 

% \hl{WE CLAIM STATE OF THE ART BELOW - Not sure what to do with this part...}
With the ability to answer SetOf and Membership queries in $O(1)$ time, we can immediately obtain faster algorithms by turning to work on sublinear-time set cover algorithms, which assumes similar query access \cite{Har-PeledIndykMahabadi:2016,IndykMahabadiRubinfeld:2018,KoufogiannakisYoung:2014}. Specifically, a modification of an algorithm from \cite{KoufogiannakisYoung:2014} can be shown to provide an $O(\log n)$ approximation to an unweighted set cover problem with $n$ elements, $m$ sets, and optimal solution $\OPT$ in time $\tilde{O}(m + n\cdot \OPT)$ \cite{Vakilian:2025}. Applying this result separately to each node, 
we can obtain an $O(\log n)$ approximation to the sparsest navigable graph in time $\tilde{O}\left(n\cdot \OPT\right)$, which is always bounded by $\tilde{O}(n^{2.5})$ since every point set has a navigable graph with $\OPT = O(n^{1.5})$ (see \cite{DiwanGouMusco:2024} or \Cref{sec:prelims}). Recall that our final goal (\Cref{thm:main_navigability}) is a $\tilde{O}(n^2)$ time algorithm, but $\tilde{O}(n^{2.5})$ is already a significant improvement on the naive $O(n^3)$. 

Similar bounds can be obtained for constructing sparse $\alpha$-shortcut reachable and $\tau$-monotonic graphs, which as discussed, are also of interest for graph-based search \cite{SubramanyaDevvritKadekodi:2019,IndykXu:2023,PengChoiChan:2023}. These problems can also be framed as collections of set cover problems, and SetOf/Membership queries can likewise be implemented in $O(1)$ time with slightly different preprocessing. Overall, we obtain the following result, formally proven in \Cref{sec:generic_applications}:
\begin{corollary}
    \label{cor:alpha-tau}
    For any $\alpha \geq 1$, let $\OPT^{\alpha}$ be the minimum number of edges in an $\alpha$-shortcut reachable graph for a point set $P$ and distance function $d$. For a fixed constant $c$, there is an algorithm that constructs an $\alpha$-shortcut reachable graph with $\leq c\log n \cdot \OPT^{\alpha}$ edges in $\tilde{O}\left(\min(n^{2.5},n\cdot \OPT^{\alpha}))\right)$ time. Similarly, for any $\tau \geq 0$, let $\OPT^{\tau}$ be the minimum number of edges in a $\tau$-monotonic graph for $P$. There is an algorithm that constructs a $\tau$-monotonic graph with $\leq c\log n \cdot \OPT^{\tau}$ edges in $\tilde{O}\left(\min(n^{2.5},n\cdot \OPT^{\tau}))\right)$ time.
\end{corollary}

\medskip \noindent \textbf{Greedy Simulation.}
More work is required to obtain the $\tilde{O}(n^2)$ time result of \Cref{thm:main_navigability}. In particular, this result takes advantage of additional structure in the set cover instances that arise from the minimum navigable graph problem. To do so, we introduce a new, purely combinatorial algorithm for set cover with SetOf/Membership queries, and then show how to accelerate our new method with careful edge preprocessing. 

% Ultimately, our nearly quadratic result will be obtained by combining a sublinear-time set cover algorithm with optimizations specialized to the set cover instances that arise when constructing a navigable graph. It may be possible to modify existing approaches, like that of \cite{KoufogiannakisYoung:2014}, which is based on solving the fractional set cover problem before applying randomized rounding. We instead introduce and modify an alternative, purely combinatorial approach to sublinear-time set cover, which we describe next. 

Concretely, our approach is based on \emph{directly simulating} the natural greedy algorithm for set cover, which repeatedly selects the set that covers the largest number of uncovered elements. This differs from the sublinear time method from \cite{KoufogiannakisYoung:2014}, which is based on solving the fractional set cover problem and rounding the solution.
To simulate the greedy method, we uniformly sample uncovered points and run SetOf queries to record what sets those points are contained in. Doing so allows us to estimate the number of uncovered elements in any set. These estimates can then be used to choose a near-optimal set at each iteration, yielding the same $O(\log n)$ approximation that is obtained when the truly optimal greedy set is chosen.
This greedy simulation is reminiscent of an approach from Nguyen and Onak \cite{NguyenOnak:2008}.
% although the runtime of their method scales exponentially with both the maximum size of any set and the maximum number of sets that any element is contained in: in our setting, both of these numbers could be as large as $n$.
It also bears high-level similarity to an algorithm of Indyk et al. \cite{IndykMahabadiRubinfeld:2018}, which samples random elements and selects multiple sets at once by finding a near-optimal set cover for those elements. 

\medskip \noindent \textbf{Edge Preprocessing.} 
% Their algorithm, however, requires different techniques based on whether \OPT{} is small or large, and uses sampling to obtain smaller sub-problems that are then solved to obtain good coverage for the original problem.
A first optimization to our greedy simulation method comes from sampling roughly $\OPT$ sets uniformly at random to add to the set cover for each node, where $\OPT$ is the size of the minimum set cover for the node.\footnote{While we do not know $\OPT$ in advance, it can be estimated via an exponential grid search.}
Doing so does not significantly impact our final approximation factor, and ensures that all elements contained in $\gtrsim  m/ \OPT$ sets are covered after this first sampling step (here $m$ is the total number of sets, which is $n-1$ for the navigability problem). 
As a consequence, the cost of recording what sets a randomly sampled element is contained in can be bounded by $\lesssim  m/ \OPT$ instead of the naive $m$. We note that the same preprocessing technique is used in the sublinear time set cover algorithm of \cite{IndykMahabadiRubinfeld:2018}. 

Already, this simple optimization allows our greedy simulation to match the efficiency of current state-of-the-art sublinear time algorithms for set cover -- see \Cref{sec:sublinear-set-cover} for more details. 
Moreover, adding random edges can lead to an even faster algorithm in the minimum navigable graph setting.  Indeed, if the optimum degree of every node is roughly the same -- i.e., $O(\OPT^*/n)$, where $\OPT^*$ is the number of total edges in the sparsest navigable graph -- one can prove that the random edges cover all but $O(n^3 / \OPT^*)$ elements across all $n$ set cover instances. The number of uncovered elements plays a direct role in the running time of our algorithm, and this bound suffices to show that the total running time for all $n$ instances is at most $\tilde{O}(n^2)$.

However, more thought is required when the optimal degrees vary (they could by a lot). In this case, we are unable to make strong claims about the absolute number of elements covered by random edges.
We deal with this challenge by introducing an additional preprocessing step: essentially, we group nodes with similar optimum degree and connect them by a clique. 
The effect of doing so is that any $p_j$ not in the clique is covered in \emph{all but one of the set cover instances corresponding to nodes in the clique}. In particular, $p_j$ is covered in all instances besides the one corresponding to its nearest neighbor in the clique. We defer details to \Cref{sec:main_algo}, but roughly speaking, this clique forming technique allows for fine-grained control over the number of uncovered elements in each instance after preprocessing, yielding our final $\tilde{O}(n^2)$ time result.

A side consequence of this final technique (discussed further in \Cref{sec:prelims}) is that it actually allows us to prove a slightly tighter existence result for navigable graphs: we show that any point set in general position admits a navigable graph with $O(n^{3/2})$ edges. Prior work achieved a bound of $O(n^{3/2}\sqrt{\log n})$ \cite{DiwanGouMusco:2024}. 
% We comment on this finding further in \Cref{sec:prelims}.

\input{prelims}
\input{mainAlgorithm}
\input{Hardness}

\section*{Acknowledgements} CM was supported by NSF Award IIS-2106888. LD and RW were supported by NSF awards CCF-2403235 and CNS-231719.

\bibliographystyle{plain}
\bibliography{references}
\appendix
\input{genericAlgorithms}

\end{document}

%% file: preamble.tex
\usepackage{soul}
\usepackage{color}
\usepackage{graphicx}
\usepackage{amsmath}
\usepackage{amsthm}
\usepackage{amsfonts}
\usepackage{amssymb}
\usepackage{tikz}
\usepackage{tabularx}
\usetikzlibrary{calc}
\usepackage{algorithm}
\usepackage[noend]{algpseudocode}
\usepackage{xcolor}
\definecolor{c0}{HTML}{641a80}
\definecolor{c1}{HTML}{b73779}
\usepackage{hyperref}
\hypersetup{
    colorlinks,
    linkcolor=c0,
    citecolor=c1,
    urlcolor=c0
}
\usepackage{placeins}
\usepackage{cleveref}
\usepackage{fullpage}
\usepackage{float}% http://ctan.org/pkg/float
\usepackage{comment}
\usepackage{xspace}

\DeclareMathOperator*{\argmin}{argmin}
\newcommand{\OPT}{\ensuremath{\mathsf{OPT}\xspace}}

\theoremstyle{plain}
\newtheorem{lemma}{Lemma}[section]
\newtheorem{observation}[lemma]{Observation}
\newtheorem{claim}[lemma]{Claim}
\newtheorem{theorem}{Theorem}
\newtheorem{corollary}[lemma]{Corollary}
\theoremstyle{definition}
\newtheorem{definition}{Definition}

\newtheorem*{rep@theorem}{\rep@title}
\newcommand{\newreptheorem}[2]{%
\newenvironment{rep#1}[1]{%
 \def\rep@title{#2 \ref{##1}}%
 \begin{rep@theorem}}%
 {\end{rep@theorem}}}
\makeatother
\newreptheorem{theorem}{Theorem}

\newcommand{\defn}[1]{\textbf{\emph{#1}}}

\algnewcommand\algorithmicforeach{\textbf{for each}}
\algdef{S}[FOR]{ForEach}[1]{\algorithmicforeach\ #1\ \algorithmicdo}
\algnewcommand\algorithmicnot{\textbf{not}}
\algdef{S}[IF]{IfNot}[1]{\algorithmicif\ \algorithmicnot\ #1\ \algorithmicthen}

\newcommand{\R}{\mathbb{R}}

\newcommand{\MemberOf}{\ensuremath{\mathrm{MemberOf}}\xspace}
\newcommand{\FreqOf}{\ensuremath{\mathrm{FreqOf}}\xspace}
\newcommand{\SetOf}{\ensuremath{\mathrm{SetOf}}\xspace}
\newcommand{\EltOf}{\ensuremath{\mathrm{EltOf}}\xspace}
\newcommand{\YES}{\text{\textsc{Yes}}\xspace}
\newcommand{\NO}{\text{\textsc{No}}\xspace}
\newcommand{\NULL}{\text{\textsc{Null}}\xspace}
\newcommand{\FAIL}{\text{\textsc{Fail}}\xspace}
\newcommand{\p}{\textsc{P}\xspace}
\newcommand{\np}{\textsc{NP}\xspace}

\newcommand{\npcomplete}{\textsc{NP-Complete}\xspace}

\newcommand{\algrule}[1][.2pt]{\par\vskip.2\baselineskip\hrule height #1\par\vskip.2\baselineskip}

%% file: bib_macros.tex
  \usepackage{nth}
  \usepackage{intcalc}

  \newcommand{\cAAAI}[1]{AAAI\ Conference\ on\ Artificial (AAAI)}

%% file: set_cover_illustration.tex
\tikzset{every picture/.style={line width=0.75pt}} %set default line width to 0.75pt        

\begin{tikzpicture}[x=0.75pt,y=0.75pt,yscale=-.8,xscale=.8]
%uncomment if require: \path (0,243); %set diagram left start at 0, and has height of 243

%Shape: Circle [id:dp44911738931133705] 
\draw   (35,146) .. controls (35,141.58) and (38.58,138) .. (43,138) .. controls (47.42,138) and (51,141.58) .. (51,146) .. controls (51,150.42) and (47.42,154) .. (43,154) .. controls (38.58,154) and (35,150.42) .. (35,146) -- cycle ;
%Shape: Circle [id:dp2926709990346378] 
\draw  [fill={rgb, 255:red, 255; green, 0; blue, 0 }  ,fill opacity=1 ] (25,55) .. controls (25,50.58) and (28.58,47) .. (33,47) .. controls (37.42,47) and (41,50.58) .. (41,55) .. controls (41,59.42) and (37.42,63) .. (33,63) .. controls (28.58,63) and (25,59.42) .. (25,55) -- cycle ;
%Shape: Circle [id:dp5006185144358287] 
\draw  [fill={rgb, 255:red, 255; green, 0; blue, 0 }  ,fill opacity=1 ] (57,52) .. controls (57,47.58) and (60.58,44) .. (65,44) .. controls (69.42,44) and (73,47.58) .. (73,52) .. controls (73,56.42) and (69.42,60) .. (65,60) .. controls (60.58,60) and (57,56.42) .. (57,52) -- cycle ;
%Shape: Circle [id:dp920491805240376] 
\draw  [fill={rgb, 255:red, 255; green, 0; blue, 0 }  ,fill opacity=1 ] (88,98) .. controls (88,93.58) and (91.58,90) .. (96,90) .. controls (100.42,90) and (104,93.58) .. (104,98) .. controls (104,102.42) and (100.42,106) .. (96,106) .. controls (91.58,106) and (88,102.42) .. (88,98) -- cycle ;
%Shape: Circle [id:dp8974795663379075] 
\draw   (85,147) .. controls (85,142.58) and (88.58,139) .. (93,139) .. controls (97.42,139) and (101,142.58) .. (101,147) .. controls (101,151.42) and (97.42,155) .. (93,155) .. controls (88.58,155) and (85,151.42) .. (85,147) -- cycle ;
%Shape: Circle [id:dp9388830392005776] 
\draw   (91,177) .. controls (91,172.58) and (94.58,169) .. (99,169) .. controls (103.42,169) and (107,172.58) .. (107,177) .. controls (107,181.42) and (103.42,185) .. (99,185) .. controls (94.58,185) and (91,181.42) .. (91,177) -- cycle ;
%Straight Lines [id:da604751545854481] 
\draw [color={rgb, 255:red, 0; green, 0; blue, 0 }  ,draw opacity=1 ][fill={rgb, 255:red, 255; green, 0; blue, 0 }  ,fill opacity=1 ][line width=1.5]    (43,146) -- (64.09,55.89) ;
\draw [shift={(65,52)}, rotate = 103.17] [fill={rgb, 255:red, 0; green, 0; blue, 0 }  ,fill opacity=1 ][line width=0.08]  [draw opacity=0] (12.48,-3.12) -- (0,0) -- (12.48,3.12) -- cycle    ;
%Shape: Circle [id:dp23812391679827216] 
\draw   (281,148) .. controls (281,143.58) and (284.58,140) .. (289,140) .. controls (293.42,140) and (297,143.58) .. (297,148) .. controls (297,152.42) and (293.42,156) .. (289,156) .. controls (284.58,156) and (281,152.42) .. (281,148) -- cycle ;
%Shape: Circle [id:dp38956457415932344] 
\draw   (271,57) .. controls (271,52.58) and (274.58,49) .. (279,49) .. controls (283.42,49) and (287,52.58) .. (287,57) .. controls (287,61.42) and (283.42,65) .. (279,65) .. controls (274.58,65) and (271,61.42) .. (271,57) -- cycle ;
%Shape: Circle [id:dp3409574225221237] 
\draw   (303,54) .. controls (303,49.58) and (306.58,46) .. (311,46) .. controls (315.42,46) and (319,49.58) .. (319,54) .. controls (319,58.42) and (315.42,62) .. (311,62) .. controls (306.58,62) and (303,58.42) .. (303,54) -- cycle ;
%Shape: Circle [id:dp4471338116147948] 
\draw  [fill={rgb, 255:red, 0; green, 0; blue, 255 }  ,fill opacity=1 ] (334,100) .. controls (334,95.58) and (337.58,92) .. (342,92) .. controls (346.42,92) and (350,95.58) .. (350,100) .. controls (350,104.42) and (346.42,108) .. (342,108) .. controls (337.58,108) and (334,104.42) .. (334,100) -- cycle ;
%Shape: Circle [id:dp9506148944326429] 
\draw  [fill={rgb, 255:red, 0; green, 0; blue, 255 }  ,fill opacity=1 ] (331,149) .. controls (331,144.58) and (334.58,141) .. (339,141) .. controls (343.42,141) and (347,144.58) .. (347,149) .. controls (347,153.42) and (343.42,157) .. (339,157) .. controls (334.58,157) and (331,153.42) .. (331,149) -- cycle ;
%Shape: Circle [id:dp2650672047853101] 
\draw  [fill={rgb, 255:red, 0; green, 0; blue, 255 }  ,fill opacity=1 ] (337,179) .. controls (337,174.58) and (340.58,171) .. (345,171) .. controls (349.42,171) and (353,174.58) .. (353,179) .. controls (353,183.42) and (349.42,187) .. (345,187) .. controls (340.58,187) and (337,183.42) .. (337,179) -- cycle ;
%Straight Lines [id:da6358306797101252] 
\draw [color={rgb, 255:red, 0; green, 0; blue, 0 }  ,draw opacity=1 ][fill={rgb, 255:red, 0; green, 0; blue, 255 }  ,fill opacity=1 ][line width=1.5]    (289,148) -- (335,148.92) ;
\draw [shift={(339,149)}, rotate = 181.15] [fill={rgb, 255:red, 0; green, 0; blue, 0 }  ,fill opacity=1 ][line width=0.08]  [draw opacity=0] (12.48,-3.12) -- (0,0) -- (12.48,3.12) -- cycle    ;
%Shape: Circle [id:dp5840737474405004] 
\draw   (512,148) .. controls (512,143.58) and (515.58,140) .. (520,140) .. controls (524.42,140) and (528,143.58) .. (528,148) .. controls (528,152.42) and (524.42,156) .. (520,156) .. controls (515.58,156) and (512,152.42) .. (512,148) -- cycle ;
%Shape: Circle [id:dp2335555558558039] 
\draw   (502,57) .. controls (502,52.58) and (505.58,49) .. (510,49) .. controls (514.42,49) and (518,52.58) .. (518,57) .. controls (518,61.42) and (514.42,65) .. (510,65) .. controls (505.58,65) and (502,61.42) .. (502,57) -- cycle ;
%Shape: Circle [id:dp3524594943150836] 
\draw   (534,54) .. controls (534,49.58) and (537.58,46) .. (542,46) .. controls (546.42,46) and (550,49.58) .. (550,54) .. controls (550,58.42) and (546.42,62) .. (542,62) .. controls (537.58,62) and (534,58.42) .. (534,54) -- cycle ;
%Shape: Circle [id:dp1112792633021743] 
\draw  [fill={rgb, 255:red, 255; green, 255; blue, 255 }  ,fill opacity=1 ] (565,100) .. controls (565,95.58) and (568.58,92) .. (573,92) .. controls (577.42,92) and (581,95.58) .. (581,100) .. controls (581,104.42) and (577.42,108) .. (573,108) .. controls (568.58,108) and (565,104.42) .. (565,100) -- cycle ;
%Shape: Circle [id:dp22495649742103996] 
\draw  [fill={rgb, 255:red, 65; green, 117; blue, 5 }  ,fill opacity=1 ] (562,149) .. controls (562,144.58) and (565.58,141) .. (570,141) .. controls (574.42,141) and (578,144.58) .. (578,149) .. controls (578,153.42) and (574.42,157) .. (570,157) .. controls (565.58,157) and (562,153.42) .. (562,149) -- cycle ;
%Shape: Circle [id:dp9885768453362245] 
\draw  [fill={rgb, 255:red, 65; green, 117; blue, 5 }  ,fill opacity=1 ] (568,179) .. controls (568,174.58) and (571.58,171) .. (576,171) .. controls (580.42,171) and (584,174.58) .. (584,179) .. controls (584,183.42) and (580.42,187) .. (576,187) .. controls (571.58,187) and (568,183.42) .. (568,179) -- cycle ;
%Straight Lines [id:da8926238499199158] 
\draw [color={rgb, 255:red, 0; green, 0; blue, 0 }  ,draw opacity=1 ][fill={rgb, 255:red, 0; green, 0; blue, 255 }  ,fill opacity=1 ][line width=1.5]    (520,148) -- (572.5,177.06) ;
\draw [shift={(576,179)}, rotate = 208.97] [fill={rgb, 255:red, 0; green, 0; blue, 0 }  ,fill opacity=1 ][line width=0.08]  [draw opacity=0] (12.48,-3.12) -- (0,0) -- (12.48,3.12) -- cycle    ;

% Text Node
\draw (17,143) node [anchor=north west][inner sep=0.75pt]   [align=left] {$\displaystyle p_{1}$};
% Text Node
\draw (106,134) node [anchor=north west][inner sep=0.75pt]   [align=left] {$\displaystyle p_{5}$};
% Text Node
\draw (1,39) node [anchor=north west][inner sep=0.75pt]   [align=left] {$\displaystyle p_{2}$};
% Text Node
\draw (60,24) node [anchor=north west][inner sep=0.75pt]   [align=left] {$\displaystyle p_{3}$};
% Text Node
\draw (102,75) node [anchor=north west][inner sep=0.75pt]   [align=left] {$\displaystyle p_{4}$};
% Text Node
\draw (106,183) node [anchor=north west][inner sep=0.75pt]   [align=left] {$\displaystyle p_{6}$};
% Text Node
\draw (94,38) node [anchor=north west][inner sep=0.75pt]  [font=\Large] [align=left] {$\displaystyle {\displaystyle \textcolor[rgb]{1,0,0}{{S}_{1\rightarrow 3}}}$};
% Text Node
\draw (263,145) node [anchor=north west][inner sep=0.75pt]   [align=left] {$\displaystyle p_{1}$};
% Text Node
\draw (352,136) node [anchor=north west][inner sep=0.75pt]   [align=left] {$\displaystyle p_{5}$};
% Text Node
\draw (247,41) node [anchor=north west][inner sep=0.75pt]   [align=left] {$\displaystyle p_{2}$};
% Text Node
\draw (306,26) node [anchor=north west][inner sep=0.75pt]   [align=left] {$\displaystyle p_{3}$};
% Text Node
\draw (348,77) node [anchor=north west][inner sep=0.75pt]   [align=left] {$\displaystyle p_{4}$};
% Text Node
\draw (352,185) node [anchor=north west][inner sep=0.75pt]   [align=left] {$\displaystyle p_{6}$};
% Text Node
\draw (355,112) node [anchor=north west][inner sep=0.75pt]  [font=\Large] [align=left] {$\displaystyle {\displaystyle \textcolor[rgb]{0,0,1}{{S}_{1\rightarrow 5}}}$};
% Text Node
\draw (494,145) node [anchor=north west][inner sep=0.75pt]   [align=left] {$\displaystyle p_{1}$};
% Text Node
\draw (583,136) node [anchor=north west][inner sep=0.75pt]   [align=left] {$\displaystyle p_{5}$};
% Text Node
\draw (478,41) node [anchor=north west][inner sep=0.75pt]   [align=left] {$\displaystyle p_{2}$};
% Text Node
\draw (537,26) node [anchor=north west][inner sep=0.75pt]   [align=left] {$\displaystyle p_{3}$};
% Text Node
\draw (579,77) node [anchor=north west][inner sep=0.75pt]   [align=left] {$\displaystyle p_{4}$};
% Text Node
\draw (583,185) node [anchor=north west][inner sep=0.75pt]   [align=left] {$\displaystyle p_{6}$};
% Text Node
\draw (596,153) node [anchor=north west][inner sep=0.75pt]  [font=\Large] [align=left] {$\displaystyle {\displaystyle \textcolor[rgb]{0.25,0.46,0.02}{{S}_{1\rightarrow 6}}}$};

\end{tikzpicture}

%% file: prelims.tex
\section{Preliminaries}
\label{sec:prelims}

\noindent \textbf{Set Cover Notation.}
Central to our approach is the set cover formulation of the sparsest navigable graph problem. We begin by reiterating notation introduced in the introduction. Formally, for each point $p_i \in P = \{p_1, \ldots, p_n\}$, we have a set cover instance $\mathcal{I}_i = ( P \setminus \{p_i\}, \mathcal{F}_i)$ with elements $P \setminus \{p_i\}$ and sets $\mathcal{F}_i = \{S_{i\rightarrow j}: j \neq i\}$, where the set $S_{i\rightarrow j}$ contains all $p_k \in  P \setminus \{p_i\}$ such that $d(p_j,p_k) < d(p_i,p_k)$. 

Our solution to the set cover instances will be a set of out-neighborhoods $N_1,\ldots, N_n$, one for each point. If $p_j \in N_i$ this indicates that set $S_{i \rightarrow j}$ was selected in our solution to $\mathcal{I}_i$. These neighborhoods comprise our final navigable graph, $G$. We let $\OPT_i$ denote the minimum size solution to set cover instance $\mathcal{I}_i$ and $\OPT^* = \sum_{i=1}^n \OPT_i$ denote the minimum number of edges in a navigable graph for point set $P$.

\medskip
\noindent \textbf{Distance Functions.} As discussed, our algorithms work with any distance function $d$, where we define:
\begin{definition}[Distance Function]
    \label{def:dist}
    A \emph{distance function} $d: P\times P\rightarrow \R^{\geq 0}$ for a point set $P = \{p_1, \ldots, p_n\}$ is any function such that $d(p_i,p_i) = 0$ for all $i$ and $d(p_j,p_i) = d(p_i,p_j) > 0$ for all $j \neq i$. 
\end{definition} 
Importantly, we do not require triangle inequality.
Throughout, we make the mild assumption that the point set $P$ is in ``general position'' meaning that 1) it contains no repeated data points and 2) for all $i, j, k$, $d(p_i, p_j) \neq d(p_i, p_k)$. These conditions are essentially always true in practice, or can be ensured by adding arbitrarily small random perturbations to the data. While assuming general position is not strictly necessary, as in prior work \cite{Al-JazzaziDiwanGou:2025}, doing so greatly simplifies the definition of navigability. In particular, \Cref{def:navigability} only captures the usual definition of navigability (that greedy search always converges for a query $q\in P$) if  there is a strict ordering on the distances $d(p_i, p_1), \ldots, d(p_i, p_n)$. Otherwise ``greedy search'' is not well-defined unless we introduce a tie-breaking rule that determines where to move if two neighbors are equally close to a query. Avoiding the notation and corner-casing required to handle tie-breaking simplifies our results. 

\medskip
\noindent \textbf{Other Notation.} Throughout, we use ``with high probability'' to mean with probability $1-1/n^c$ for an arbitrarily large constant $c$ (the specific choice of which will impact other constant terms in the statement).

\subsection{The Distance-Based Permutation Matrix and Fast Query Access}
As in work on existential results for navigability \cite{DiwanGouMusco:2024}, it will be helpful to think about our navigability problem as represented in a ``distance-based permutation matrix'', $\Pi$, with $n$ rows and $n$ columns. The $i^\text{th}$ row contains the points in $P$ sorted in increasing order of their distance from $p_i$. See \Cref{fig:distance_based_permutation_matrix} for an example. The value of the distance-based permutation matrix is that it only takes $O(n^2\log n)$ time to construct, yet allows us to efficiently query all $n$ set cover instances $\mathcal{I}_1, \ldots, \mathcal{I}_n$. In particular, we have the following:

\begin{figure}[h]
\centering
\vspace{-.5em}
\includegraphics[width=.45\textwidth]{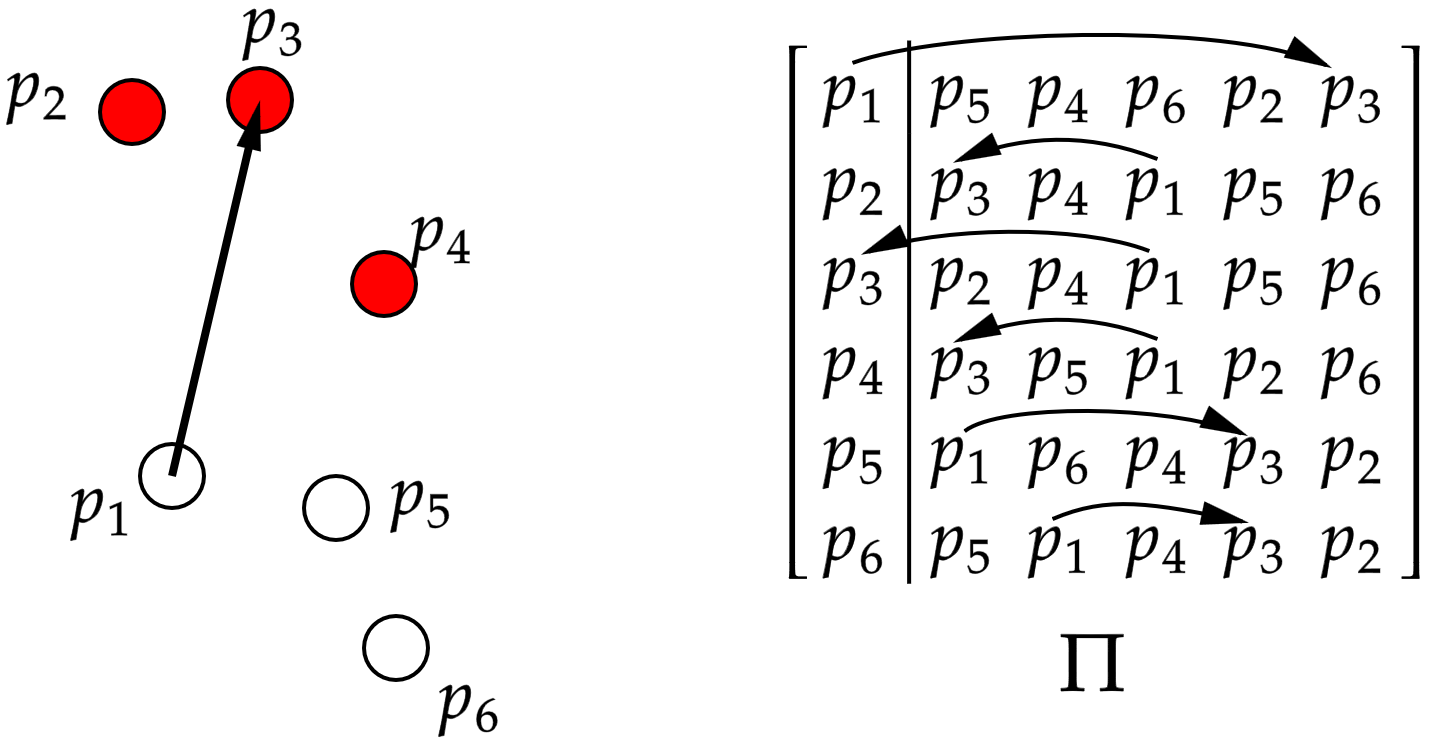}
\vspace{-.5em}
\caption{\label{fig:distance_based_permutation_matrix}An example distance-based permutation matrix, $\Pi$, for a data set in two-dimensional Euclidean space.
The $i^\text{th}$ row contains all points in $P$ sorted in increasing order of their distance from $p_i$. 
 $\Pi$'s first row is $[p_1, p_5, p_4, p_6, p_2 p_3]$ since $d(p_1,p_1) < d(p_1,p_5) <\ldots < d(p_1,p_3)$.
$\Pi$ can be used to quickly identify what sets a node $p_k$ is contained in for a particular set cover instance $\mathcal{I}_i$: if $p_k \in S_{i\rightarrow j}$, then $p_j$ will lie to the \emph{left} of $p_i$ in the $k^\text{th}$ row of $\Pi$. This property is illustrated for a particular set $S_{1\rightarrow 3} = \{p_2, p_3, p_4\}$, highlighted in red. As we can see, $p_3$ lies to the left of $p_1$ in rows $2,3,4$ in $\Pi$, and to the right in all other rows.} 
    \vspace{-.75em}
\end{figure}

\begin{claim}\label{clm:sublinear_access}
Given a set of $n$ points, $P$, and a distance function $d$ that can be evaluated in $T$ time for any $p_i,p_j \in P$, we can construct in $O(n^2T + n^2\log n)$ time a data structure that answers the following queries in constant time for any $p_i,p_j,p_k \in \{1,\ldots, n\}$
\begin{enumerate}
    \item $\MemberOf\left(p_i,p_j,p_k\right)$: Returns \YES if $p_k\in S_{i\to j}$, \NO otherwise.
    \item $\FreqOf(p_i,p_k)$: Returns the number of sets in $\mathcal{I}_i$ that $p_k$ is a member of, i.e., $|\{j: p_k \in S_{i\rightarrow j}\}|$
    \item $\SetOf(p_i,p_k,\ell)$: Returns the $\ell^\text{th}$ set (according to some fixed order) that $p_k$ is a member of in $\mathcal{I}_i$, or $\text{\NULL}$ if $\ell > \FreqOf(p_i,p_k)$. In other words, return the $\ell^\text{th}$ element of $\{j: p_k \in S_{i\rightarrow j}\}$.
\end{enumerate}
\end{claim}
Our algorithms will only ever use \SetOf queries to return all values of $j$ such that $p_k \in S_{i\rightarrow j}$. This can be done in linear time by repeatedly issuing \SetOf queries for $\ell = 1, 2, \ldots$ until a \NULL value is returned.
\begin{proof}\textbf{Data Structure.} We first compute and store all pairwise distances in $O(n^2 T)$ time. We then sort $n$ lists to build the distance-based permutation matrix, $\Pi$ in  $O(n^2\log n)$ time. We then precompute and store the rank of each $p_j$ in the $i^\text{th}$ row of $\Pi$, for all $i$. We denote this rank by $\pi_i^{-1}(j)$.  This takes $O(n^2)$ time.

\smallskip
\noindent\textbf{Answering Queries.}
To answer a \MemberOf query, we simply return \YES if $\pi_k^{-1}(j) < \pi_k^{-1}(i)$ -- i.e., $p_j$ precedes $p_i$ in the $k^\text{th}$ row of $\Pi$. We return \NO otherwise. 
To answer a \FreqOf query, we return the number of items that precede $p_i$ in row $k$, i.e. $\pi_k^{-1}(i) - 1$.
To answer a \SetOf query, we return \NULL if $\ell \geq \pi_k^{-1}(i)$. Otherwise, we return $\Pi_{k,\ell}$, the $\ell^\text{th}$ entry in the $k^\text{th}$ row of $\Pi$. All operations take $O(1)$ time.
\end{proof}

\subsection{Pre-processing via Random Edges and Arbitrary Cliques}
\label{subsec:preprocess_overview}
As discussed in \Cref{sec:technical_overview}, a key component of our algorithm is to reduce the cost of running a greedy set cover method by ``pre-processing'' the navigability problem by adding a small number of less carefully chosen edges that nevertheless aid in covering many elements. We do so in two ways: by adding random edges and adding random cliques. We discuss basic properties of these strategies below. 

\subsubsection{The Power of Random Edges}\label{subsec:prelims-random-edges} A first observation, which is also used in prior work on sublinear time set cover \cite{IndykMahabadiRubinfeld:2018}, is that adding random sets to instances $\mathcal{I}_1, \ldots, \mathcal{I}_n$ will cover all elements contained in many sets, with high probability. In our setting, this corresponds to adding random out-neighbors to each node. We have the following basic claim:

\begin{claim}
    \label{clm:random_edges}
Suppose we sample $\beta$ indices $j_1, \ldots, j_\beta$ uniformly at random from $\{1, \ldots, n\}\setminus i$ and add $S_{i \to j_1}, \ldots, S_{i \to j_\beta}$ to our solution, $N_i$, for set cover instance $\mathcal{I}_i$. Then for any $p_k$ contained in at least $c\ln(n)\cdot \frac{n}{\beta}$ sets $S_{i\rightarrow j}$, we have that $p_k \in S_{i \to j_1}\cup\cdots \cup S_{i \to j_\beta}$ with probability $> 1 - 1/n^c$. 
\end{claim}
\begin{proof}
    We can assume $S_{i \to j_1}, \ldots, S_{i \to j_\beta}$ are drawn with replacement, as the probability that $p_k$ is covered is only higher if they are drawn without replacement.
    Since $p_k$ is contained in $z \geq cn \ln (n)/\beta$ sets, for any randomly-chosen $\ell$, $\Pr[p_k\in S_{i \to \ell}] = \frac{z}{n-1}$. Accordingly, we can bound:
    \[
    \Pr[p_k \notin  S_{i \to j_1}\cup\ldots, S_{i \to j_\beta}] \leq \left(1 - \frac{z}{n-1}\right)^{\beta} < \left(1 - \frac{c\ln n}{\beta}\right)^{\beta} \leq \left(\frac{1}{e}\right)^{c\ln n} =  \frac{1}{n^c}.\qedhere
    \]
\end{proof}
An implication of \Cref{clm:random_edges} is that, if we add $\beta$ random edges to every node in our graph, then with high probability, for all instances $\mathcal{I}_i$, the only elements remaining to be covered are contained in $\leq n/\beta$ sets. This property will be critical in our efficient simulation of the greedy set cover algorithm, which requires sampling a random uncovered element and recording all sets it is contained in.

Due to the additional structure of the navigable graph problem, adding random sets actually guarantees another property that does not hold for arbitrary collections of set cover instances: we can ensure that the total number of uncovered elements across all set cover instances is significantly reduced. 
\begin{corollary}
    \label{cor:low_total_points}
Suppose that for each set cover instance $\mathcal{I}_i$, as in \Cref{clm:random_edges}, we add $\beta$ sets uniformly at random to our solution, $N_i$. For each $i\in[n]$, let $U_i$ denote the set of all uncovered points $p_k$ that are not contained in any of the random sets chosen for instance $\mathcal{I}_i$. Then with probability $> 1 - 1/n^{c-2}$, 
\begin{align*}
    \sum_{i=1}^n |U_i| < c\log n \cdot \frac{n^2}{\beta}.
\end{align*}
\end{corollary}\label{cor:total-uncovered-points}
% In particular, if we choose $\beta = {O}(\OPT^*\log (n)/n)$, then we have at most $\leq n^3/OPT^*$ elements remaining to be covered  across all set cover instances after adding the random edges. 
\begin{proof}
By \Cref{clm:random_edges} and a union bound, we have that, simultaneously for all $i$, all $p_k$ contained in $\geq c\log(n)\cdot \frac{n}{\beta}$ sets in $\mathcal{I}_i$ are in at least one of the randomly-chosen sets for instance $i$. So $U_i$ only contains elements in $< c\log(n)\cdot \frac{n}{\beta}$ sets. Referring to the distance-based permutation matrix, we notice that there are $n\cdot  (c\log(n)\cdot \frac{n}{\beta} -2)$ such elements across all instances: these elements are exactly those contained in columns proceeding column $c\log(n)\cdot \frac{n}{\beta}$ in the matrix (excluding the first column). It follows that $\sum_{i=1}^n |U_i| < c\log n \cdot \frac{n^2}{\beta}.$
\end{proof}

\subsubsection{The Power of Cliques}\label{subsec:prelims-cliques} Interestingly, a result similar to \Cref{cor:low_total_points} can be obtained without randomness, but rather by adding arbitrary cliques to our graph. In particular, we can argue that, if we add a clique over a subset of nodes, $K$, then the total number of uncovered elements across all $\mathcal{I}_i$ for $i \in K$ is bounded. Formally, we have:

\begin{claim}\label{lem:cliques}
    Suppose we select points $K \subset P$ and add a clique over those points to $G$. I.e., for each $p_i \in {K}$, we add $K\setminus \{p_i\}$ to $p_i$'s out-neighborhood, $N_i$. Then for any $p_k\notin K$, for all but possibly a single $p_i \in K$, $p_k \in {S}_{i\rightarrow j}$ for some $p_j \in N_i$. That is, $p_k$ is covered in all set cover instances corresponding to nodes in the clique, except at most one. Let $U_i$ contain all $p_k$ not covered in instance $\mathcal{I}_i$ by edges in $G$. It follows that:
    \begin{align*}
    \sum_{i\in K} |U_i| \leq n - |K|.
    \end{align*}
\end{claim}
\begin{proof}

Let $p_\ell = \argmin_{p_i \in K} d(p_i,p_k)$. For all $p_i \in K$, $i \neq \ell$ we have that $p_k$ is contained in $S_{i\rightarrow \ell}$, which we added when forming the clique. So $p_k$ can only be uncovered in instance $\mathcal{I}_\ell$. There are $n-|K|$ points not in the clique, hence at most $n-|K|$ uncovered elements across all instances.
\end{proof}
An immediate consequence of \Cref{lem:cliques} is that, if we partition our points into $n/\beta$ sets of size $\beta$, and add a clique for each of those sets, then we only add $O(n\beta)$ edges to the graph, but cover all but $(n-\beta)\cdot n/\beta <n^2/\beta$ elements across all $n$ set cover instances. This matches what was obtained above for random edge deletions. We will use a slightly finer-grained version of the statement in our proof of \Cref{thm:main_navigability} in \Cref{sec:main_algo}.

\begin{figure}[h]
\centering
\vspace{-.5em}
\includegraphics[width=.5\textwidth]{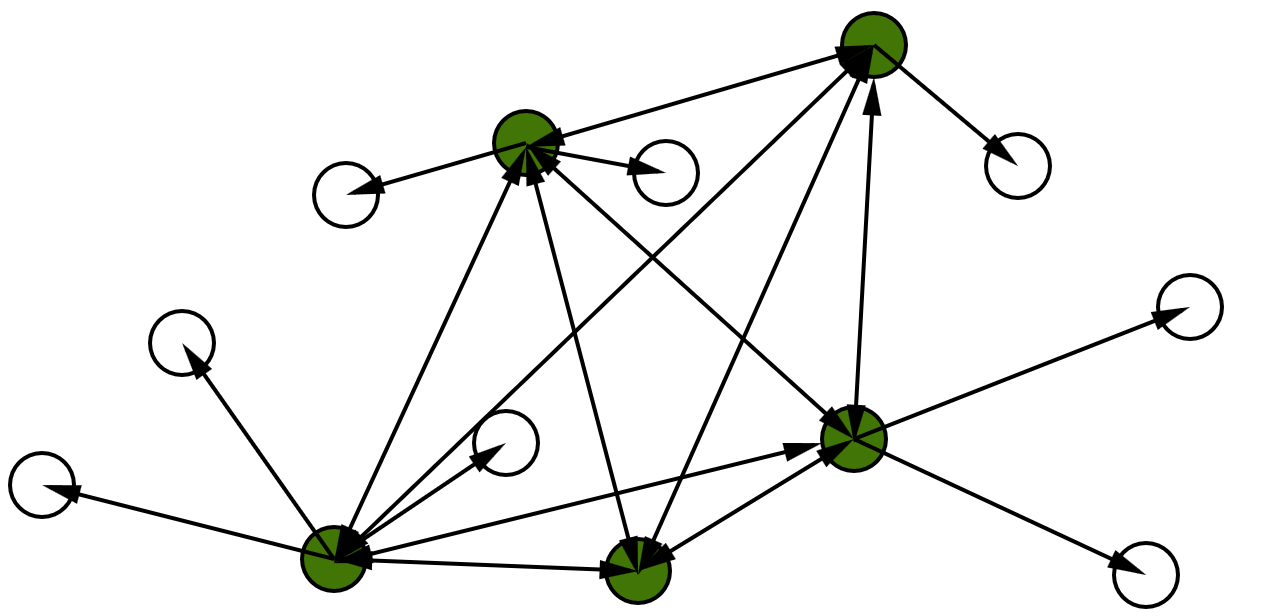}
\vspace{-.5em}
\caption{A partial illustration of the graph construction used to prove \Cref{thm:opt-upper-bound}. We begin by partitioning our points into $O(\sqrt{n})$ arbitrary groups of size $O(\sqrt{n})$. An example group, $K$, is illustrated in green. We connect all points in $K$ via a clique and for every point $p_k \notin K$, we draw an edge from $p_k$'s nearest neighbor in $K$ to $p_k$. By \Cref{lem:cliques}, this yields a navigable graph. The total number of edges is $O(n^{3/2})$.}
\label{fig:clique}
\vspace{-.75em}
\end{figure}

While a fairly simple observation, the argument above can actually be used to tighten the previous best upper bound of $2n^{3/2}\sqrt{\ln n}$ on the number of edges in the minimum navigable graph for an arbitrary distance function and point set (in general position) \cite{DiwanGouMusco:2024}. In particular, simply partition the points into $\lceil\sqrt{n}\rceil$ arbitrary groups of size $\leq \sqrt{n}$. For each group $K$, add a clique between the nodes in that group. For any $p_k \notin K$, add an edge from $p_\ell = \argmin_{p_i \in K} d(p_i,p_k)$. 
This yields a navigable graph, and we add at most $|K|(|K|-1) + n - |K| = |K|^2 + n - 2|K|$ edges per group $K$, for a total of $\left(2n - 2\sqrt{n}\right)\cdot \sqrt{n}$ edges when $n$ is a perfect square. It can be verified that essentially the same bound holds for non-square $n$, yielding:
\begin{theorem}\label{thm:opt-upper-bound}
    There is a deterministic algorithm that, given any point set $P$ of size $n$ and distance function $d$, constructs a navigable graph on $P$ with at most $2n^{3/2} - n$ edges in $O(n^2)$ time.
\end{theorem}

%% file: mainAlgorithm.tex
\newcommand{\NGCover}{\textsc{VoteCover}}

\section{A Nearly Quadratic Time Algorithm}
\label{sec:main_algo}
In this section, we prove \Cref{thm:main_navigability} by presenting an algorithm that returns an $O(\log n)$-approximation to the minimum navigable graph for any set of $n$ points and any distance function in $\tilde{O}(n^2)$ time.

\newcommand{\thres}{\ensuremath{\gamma}\xspace}

\subsection{Vote Cover Subroutine}
Our algorithm is based on efficiently simulating the standard greedy algorithm for each of the $n$ set cover instances that corresponds to the navigable graph problem (see \Cref{sec:prelims}). Our general simulation routine, which we call \NGCover{}, is detailed in \Cref{alg:ng_cover}. \NGCover{} takes as input a source point $p_i \in P$ and a set of points $U \subseteq P\setminus \{p_i\}$ to be covered. 
It returns a set of neighbors for $p_i$ that cover $U$. The reason we specify a set of points $U$ instead of asking to cover all of $P\setminus \{p_i\}$ is because \NGCover{} will ultimately be combined with ``edge preprocessing'' techniques (discussed in \Cref{sec:tech_overview}) that cover some points in advance to reduce runtime.
The algorithm also takes in a positive parameter $\thres$ that controls its success probability, as well as an optional size limit $\ell$, which specifies an early-termination condition; that is, we allow \NGCover{} to return \FAIL in lieu of a cover of size $> 2\ell$. Early termination will also be helpful when combining \NGCover{} with the preprocessing techniques. Our main result on the algorithm is as follows:

% As hinted at in Section~\ref{sec:technical_overview}, our approach to constructing a navigable graph is to initialize $n$ independent set cover instances, each responsible for the neighborhood of a particular point. We then run a subroutine on each of these instances, which we call \NGCover{}, which computes for each point a navigable neighborhood. We start by presenting this subroutine.

\newcommand{\VT}{15}
\newcommand{\CA}{30}
\newcommand{\HF}{4}
\newcommand{\SumCAVT}{\number\numexpr2*\VT+\CA\relax}
\ifnum\numexpr2*\VT\relax>\CA
    \newcommand{\MaxCAVT}{\number\numexpr2*\VT\relax}
\else
    \newcommand{\MaxCAVT}{\CA}
\fi
\newcommand{\VoteThres}{\thres \ln n}
\newcommand{\MaxLimit}{\MaxCAVT \cdot \OPT{}_i \ln n}
\newcommand{\CoverSize}{\SumCAVT \OPT{}_i \ln n}
\newcommand{\CoverCSize}{\CA \cdot \OPT{}_i \ln n}
\newcommand{\CoverVSize}{\number\numexpr2*\VT\relax\cdot \OPT{}_i \ln n}

% Our efficient construction of sparse navigable graphs is contingent on being able to efficiently approximate a minimum set cover for an instance $\mathcal{I}_i$.

\begin{theorem}
    \label{thm:ng_cover_analysis}
    Given a point set $P$ preprocessed as in Claim~\ref{clm:sublinear_access}, a source point $p_i \in P$, 
    a subset of points $U \subseteq P \setminus \{p_i\}$ to cover, a size limit $\ell$, and parameter $\thres$, \Cref{alg:ng_cover} returns either $\FAIL$ or a set 
    $Q \subseteq P$ with $|Q| \leq \ell$ such that, for any $p_k \in U$, there is some $p_j \in Q$ satisfying $p_k \in S_{i \to j}$. 
    The algorithm runs in $\tilde O((f + |U|) \cdot \min(\ell, \thres\OPT{}_i))$, where $f$ is the maximum number of sets $S_{i \to j}$ that any $p_k \in U$ is contained in and $\OPT_i$ is the optimal cover size for set cover instance $\mathcal{I}_i$, as defined in \Cref{sec:prelims}. Moreover, if $\ell \geq \min(n, (2\thres + 5)\OPT_i\ln n)$, then \Cref{alg:ng_cover} successfully returns a set with probability $\geq 1 - 1/n^{\thres/4 - 2}$.
    % , where $\thres$ is a constant hyperparameter of the algorithm. 
\end{theorem}
%  Given these operations, we now provide a detailed discussion of \NGCover{}, whose pseudocode we provide in \Cref{alg:ng_cover}. 

\begin{algorithm}[htb]
    \caption{VoteCover($P, p_i, U \subseteq P, \ell$)}
    \label{alg:mng-cover}
    \begin{algorithmic}[1]
        \Require A dataset $P$ pre-processed as in \Cref{clm:sublinear_access}, a source point $p_i$, a set $U\subseteq P$ of uncovered elements in set cover instance $\mathcal{I}_i$, a size limit $\ell$, and a parameter $\thres > 0$ that controls success probability.
        \Ensure A set of at most $\ell$ neighbors for $p_i$ that cover $U$, or $\FAIL$.
        \algrule
        \State Initialize the cover $C \gets \varnothing$ and the voter pool $V \gets \varnothing$. 
        \State Initialize $votes \gets$ array of $n$ empty sets.
        \While{$|U| > 0$}
            \State Select uniformly random $p_k$ from $U$. Set $U \gets U \setminus \{p_k\}$, $V \gets V \cup \{p_k\}$ \label{alg:mng-cover:sample-voters-start}\Comment{Promote a random uncovered point $p_k$ to the voter pool.}
            \State \textbf{if} $|V| + |C| > \ell$ \textbf{then} \Return \textsc{Fail}\label{alg:mng-cover:sample-voters-end}\label{alg:mng-cover:v-limit-reached}
            \ForEach{$S_{i \to j}$ containing $p_k$}\label{alg:mng-cover:cast-votes-loop}
                \State $votes[j] \gets votes[j] \cup \{p_k\}$\label{alg:mng-cover:cast-votes}\Comment{$p_k$ votes for all $S_{i \to j}$ containing it.}
                \If{$|votes[j]| \geq \VoteThres$}\label{alg:mng-cover:enough-votes-cond}
                    \State $C \gets C \cup \{p_j\}$\label{alg:mng-cover:enough-votes-start}\Comment{Add $S_{i \to j}$ to the cover if it receives enough votes.}
                    \State \textbf{if} $|V| + |C| > \ell$ \textbf{then} \Return \textsc{Fail}\label{alg:mng-cover:c-limit-reached}
                    \ForEach{$p_v \in votes[j]$}\label{alg:mng-cover:enough-votes:remove-voters}\Comment{Remove votes of all voters for $S_{i \to j}$, which are now covered.}
                        \ForEach{$S_{i \to u}$ containing $p_v$}\label{alg:mng-cover:enough-votes:undo-votes}
                            \State $votes[u] \gets votes[u] \setminus \{p_v\}$\label{alg:mng-cover:undo-vote}
                        \EndFor
                        \State $V \gets V \setminus \{p_u\}$\label{alg:mng-cover:remove-from-v}
                    \EndFor\label{alg:mng-cover:enough-votes:remove-voters-end}
                    \ForEach{$p_v \in U$}\label{alg:mng-cover:enough-votes:remove-nominees-start}\Comment{Remove all points covered by $S_{i \to j}$ from $U$.}
                        \If{$p_v \in S_{i \to j}$}
                            \State $U \gets U \setminus \{p_v\}$\label{alg:mng-cover:remove-from-n}
                        \EndIf
                    \EndFor\label{alg:mng-cover:enough-votes:remove-nominees-end}
                    \State \textbf{break}\label{alg:mng-cover:enough-votes-end}
                \EndIf
            \EndFor
        \EndWhile
        \State \Return $C \cup V$\label{alg:mng-cover:return-cover}
    \end{algorithmic}
    \label{alg:ng_cover}
\end{algorithm}
% Before proving this result, we describe the \NGCover{} (\Cref{alg:ng_cover}) in detail. 
% Recall that when $P$ is preprocessed as in Claim~\ref{clm:random_edges}, $\MemberOf$, $\FreqOf$, and $\SetOf$ can be answered in $O(1)$ time, but we do not have explicit access to the sets $S_{i \to j}$, $j \neq i$.
The standard greedy set cover algorithm initializes an empty cover and sequentially adds the set that covers the most uncovered points. Since we do not have explicit access to sets in our setting and cannot efficiently compute how many elements they cover, \NGCover{} (\Cref{alg:mng-cover}) approximates this process by randomly sampling ``voters'' from all uncovered points, which ``vote'' for all sets $S_{i \to j}$ that cover them. 
% chris up to here
% It is easy to show (see Claim \ref{clm:mng-cover:high-coverage}) that choosing the set with the most votes at (which may not be the absolute best set) suffices to obtain the same $O(\log n)$ approximation factor as the standard greedy algorithm.

In the pseudocode, we use $C$ to represent the cover, which formally contains out-neighbors, $p_j$, to add to $p_i$. If $p_j \in C$, then we have selected set $\mathcal{S}_{i\rightarrow j}$ for set cover instance $\mathcal{I}_i$. We use $V$ to denote the set of uncovered points that have been chosen as voters. 
To select a high-coverage set, we sample an uncovered point from the set $U$, move it to $V$, and have it vote once for every set which contains it. These sets can be efficiently obtained via $\SetOf$ queries. We record the individual votes in an array called $votes$, where $votes[j]$ is a set holding the ids of all voters currently voting for $S_{i \to j}$, i.e., $votes[j] = V \cap S_{i \to j}$.

When repeating this process many times, we expect the sets containing the most uncovered points to receive the most votes. The number of voters we need to sample in order for this to be a likely event depends on the amount of coverage these sets offer, so rather than choose a fixed number of samples, we instead sample points until some set reaches a threshold of $\thres\ln n$ votes. At this point, we are confident that it has near maximal coverage among all remaining sets, so we add the set to $C$. $\thres$ is a constant whose effect will become apparent in our analysis, but in brief, a larger choice of $\thres$ loosens our upper bound on the size of the returned set in exchange for increasing the probability that the set satisfies that bound. 

Rather than restart the voting from scratch each time we choose a set, which would be costly, we continue from the current state, keeping all votes except those cast by voters that are now covered. Since all newly covered voters are contained in $votes[j]$, they are easy to identify and remove. We also update our set of uncovered points, $U$, using $\MemberOf$ queries to identify and remove all points covered by the chosen set.
% We show in \Cref{clm:uniform-random-sample} that continuing in this way does not compromise the uniform random nature of our sample of voters.

We then continue to add highly-voted sets to $C$ until $U$ is empty. At this point, there may still be active voters in $V$ whose sets never received enough votes to be chosen. We finish by explicitly covering each of these voters with the set corresponding to itself. We will prove that the number of such sets we need is small.

\subsubsection{Correctness and Approximation Factor}\label{sec:votecover-correctness}
We begin by analyzing the accuracy of \Cref{alg:ng_cover}. Specifically, we show that if $\ell$ is chosen sufficiently large, then the algorithm returns a cover of size $O(\OPT{}_i \log n)$ with high probability.

\begin{lemma}\label{thm:mng-cover:logn-approx}
\NGCover{}$(P, p_i, U, \ell,\thres)$ either returns $\FAIL$ or a cover of $U$ with size $\leq \ell$. Furthermore, if $\ell \geq \min(n, (2\thres + 5)\OPT_i\ln n)$, then it returns a cover with probability $\geq 1 - 1/n^{\thres/4 - 2}$
\end{lemma}

\noindent\textit{Proof outline.} We first show that, if the failure conditions on Lines~\ref{alg:mng-cover:v-limit-reached} and \ref{alg:mng-cover:c-limit-reached} are never reached, the algorithm returns a cover of $U$ with size $O(\OPT{}_i \log n)$, with high probability. To do so, we prove that $V$ always contains a uniform random sample of the remaining uncovered points. As discussed above, this allows us to argue that any $p_j$ satisfying the condition on Line~\ref{alg:mng-cover:enough-votes-start} covers nearly as many uncovered points as the best choice of set. Adapting the standard analysis of greedy set cover, we conclude that the algorithm terminates after adding at most $O(\OPT{}_i \log n)$ sets to $C$. In addition, we show that the size of $V$ never exceeds $O(\OPT{}_i \log n)$ with high probability, so $C \cup V$ must have size at most $O(\OPT{}_i \log n)$, proving the first claim of the theorem. \\
\\
\noindent For the formal proof, we begin with some basic definitions and facts.
% Finally, we prove the second claim of the theorem by using this first result to argue that the algorithm must terminate before the size of $C$ exceeds either $\CA \cdot \OPT{}_i \ln n$ or $n$.

\begin{definition}\label{def:mng-cover:uncovered-points}
    Given an initial set $U \subseteq P\setminus \{p_i\}$ of points to cover and a set $C \subseteq P\setminus \{p_i\}$ corresponding to a partial cover, let $U_{\overline{C}}$ denote $U \setminus \bigcup_{p_c \in C} S_{i \to c}$. I.e., $U_{\overline{C}}$ contains all points not covered by $C$.
\end{definition}

% \noindent Intuitively, $U_{\overline{C}}$ indicates the remaining work to be done in covering $U$. For example, 
At the start of the algorithm, when $C = \varnothing$, $U_{\overline{C}} = U$. Each additional neighbor added to $C$ shrinks the size of $U_{\overline{C}}$. The task of \NGCover{} is to find a small $C$ such that $U_{\overline{C}} = \varnothing$.

\begin{claim}\label{clm:mng-cover:uncovered-points}
    \Cref{alg:mng-cover} maintains the invariant that $U$ and $V$ are disjoint sets such that $U \cup V = U_{\overline{C}}$.
\end{claim}

\begin{proof}
    This is clearly true at the start of the algorithm when $V = \varnothing$, $C = \varnothing$, and thereby $U_{\overline{C}} = U$. As the algorithm executes, elements are never added to $U$ or $V$ other than during a movement from $U$ to $V$ on Line~\ref{alg:mng-cover:sample-voters-start}, which has no effect $U \cup V$. Finally, points are only every removed from $U$ or $V$ when the set corresponding to some out-neighbor, $p_j$, is added to $C$. In this case, we can see that exactly those points covered by $p_j$ are removed from $U$ (Lines~\ref{alg:mng-cover:enough-votes:remove-nominees-start}-\ref{alg:mng-cover:enough-votes:remove-nominees-end}) and $V$ (Lines \ref{alg:mng-cover:enough-votes:remove-voters}-\ref{alg:mng-cover:remove-from-v}). Thus the invariant is maintained.    
\end{proof}

\begin{lemma}\label{lem:mng-cover:correctness}
    If Algorithm~\ref{alg:mng-cover} does not return \FAIL, then it returns a cover of $U$.
\end{lemma}

\begin{proof}
    If Algorithm~\ref{alg:mng-cover} is never interrupted by either failure condition on Lines~\ref{alg:mng-cover:v-limit-reached} or \ref{alg:mng-cover:c-limit-reached}, it can only return on Line~\ref{alg:mng-cover:return-cover}. This line is reached under the condition that $|U| = 0$. Since each iteration of the main loop removes at least one element from $U$, and elements are never added to $U$, we will eventually terminate. Morever, it is clear that $C \cup V$ is a valid cover of $U$ when the algorithm terminates. In particular, by \Cref{clm:mng-cover:uncovered-points}, if $|U| = 0$, then $U_{\overline{C}} = V$. I.e., the only points not covered by $C$ are those in $V$. Since $d(p_k, p_k) = 0 < d(p_i, p_k)$ for any $p_k \in V$, adding these points to the cover ensures they are also covered. 
    % We now show that . We first note that under distance functions as in \Cref{def:dist}, $p_k \in S_{i \to k}$ for any $i \not= k$, since $d(p_k, p_k) = 0 < d(p_i, p_k)$. That is, $U_{\overline{C \cup U_{\overline{C}}}} = \varnothing$ for any $C \subseteq P$. By \Cref{clm:mng-cover:uncovered-points}, it is always the case that $N \cup V = U_{\overline{C}}$ between iterations of the main loop, including after the final iteration. Moreover, $N = \varnothing$ by Line~\ref{alg:mng-cover:return-cover}, so $V = U_{\overline{C}}$. It follows that at the moment of return, $U_{\overline{C \cup V}} = U_{\overline{C \cup U_{\overline{C}}}} = \varnothing$. Therefore, the returned set, $C \cup V$, is a cover of $U$.
\end{proof}

\begin{claim}\label{clm:uniform-random-sample}
    \Cref{alg:mng-cover} maintains the invariant that $V$ is a uniform sample (without replacement) of $U_{\overline{C}}$.
\end{claim}

\begin{proof}
    This is trivially true at the start of the algorithm when $V = \varnothing$. During the algorithm, the only additions to $V$ are randomly chosen points from $U$, which by \Cref{clm:mng-cover:uncovered-points}, equals $U_{\overline{C}} \setminus V$. Thus, these random additions preserve $V$ being a random sample of $U_{\overline{C}}$. On the other hand, points are only removed from $V$ when some $p_j$ is added to $C$. The points removed are exactly those in $S_{i \to j}$. If $V$ is uniform sample of $U_{\overline{C}}$, then $V \setminus S_{i \to j}$ is a uniform random sample of $U_{\overline{C}} \setminus S_{i \to j} = U_{\overline{C \cup \{p_j\}}}$, so the invariant continues to hold.
\end{proof}

\begin{claim}\label{obs:exists-opt}
    For any $C \subseteq P$, there is always some $S_{i \to j}$ that contains at least a $\frac{1}{\OPT{}_i}$ fraction of $U_{\overline{C}}$.
\end{claim}

\begin{proof}
    Let $C_{\OPT{}}$ be an optimal cover for $P \setminus \{p_i\}$. We have that $U_{\overline{C}} \subseteq \bigcup_{p_j \in C_{\OPT{}}} (S_{j \to c} \cap U_{\overline{C}})$, which implies that $\sum_{p_j \in C_{OPT}}|S_{i \to j} \cap U_{\overline{C}}| \geq |U_{\overline{C}}|$. Some $p_j \in C_{OPT}$ must therefore satisfy $|S_{j \to c} \cap U_{\overline{C}}| \geq \frac{1}{\OPT{}_i}|U_{\overline{C}}|$.
\end{proof}

\begin{lemma}\label{lem:mng-cover:v-size}
    With probability $\geq 1 - 1/n^{\thres / 4 - 1}$, $|V|$ never exceeds $2 \thres \OPT_i \ln n$.
\end{lemma}

\begin{proof}
    By \Cref{obs:exists-opt}, there is always some set $S^* = S_{i \to j^*}$ containing at least a $\frac{1}{\OPT{}_i}$ fraction of $U_{\overline{C}}$, for any $C \subseteq P$. Thus, any randomly-selected point from $U_{\overline{C}}$ has at least a $\frac{1}{\OPT{}_i}$ probability of being contained in $S^*$. 
    If our random sample $V \subseteq U_{\overline{C}}$ contains $2 \thres \OPT_i \ln n$ points, then by a standard Chernoff bound (e.g., \cite{Wajc:2017}), $S^*$ will contain at least $\thres \ln n$ (the vote threshold from \Cref{alg:mng-cover:enough-votes-cond}) points from $V$ with probability at least:
    \[ % Bernoulli Chernoff bound
        \Pr[|S^* \cap V| \geq \thres\ln n] = 1 - \Pr[|S^* \cap V| < (1/2)(2\thres\ln n)] \geq 1 - e^{-(1/2)^2\thres\ln n} \geq 1 - 1/n^{\thres / 4}
    \] 
    % \[ % Hypergeometric tail bound - don't use; deteriorates with low probabilities
    %     \Pr\left[X \geq \left(\frac{1}{\OPT_i} - \frac{1}{2\OPT_i}\right)(2\thres\OPT_i\ln n)\right] \geq 1 - e^{-2(2\thres\OPT_i\ln n)/(2\OPT_i)^2} \geq 1 - e^{-(\thres\ln n)/\OPT_i} 1 - 1/n^{\thres/\OPT_i}
    % \]
    That is, with high probability, $S^*$ reaches the vote threshold by the time there are $2 \thres \OPT_i \ln n$ active voters. Subsequently, its voters are removed from $V$. Taking a union bound over at most $n$ states of $V$, we see that with probability $\geq 1 - 1 / n^{\thres/4 - 1}$, $V$ contains at most $2 \thres \OPT_i \ln n$ points at any one time.
\end{proof}

\begin{claim}\label{clm:mng-cover:high-coverage}
    With probability $\geq 1 - 1/n^{\thres/4 - 2}$, every set added to $C$ contains at least a $\frac{1}{5\OPT_i}$ fraction of $U_{\overline{C}}$.
\end{claim}

\begin{proof}
    When the $|V| \leq 2\thres\OPT_i\ln n$, the probability that a set $S_{small} = S_{i \to j_{small}}$ containing less than a $\frac{1}{5\OPT_i}$ fraction of $U_{\overline{C}}$ receives more than $\thres \ln n$ votes can be bounded, again by a Chernoff bound, by
    \begin{align*}
        \Pr[|S_{small} \cap V| \geq \thres \ln n] & 
        % = 1 - \Pr[|S_{small} \cap V| > (1 + 3/2)((2/5)\thres\ln n)]\\
        \leq e^{-(3/2)^2(4/35)\thres\ln n} < 1/n^{\thres/4}.
    \end{align*}
    Taking a union bound over at most $n - 1$ sets and across at most $n$ states of $V$ yields that all sets containing less than a $\frac{1}{5\OPT_i}$ fraction of $U_{\overline{C}}$ simultaneously fall short of the required $\thres \ln n$ votes with probability at least $1 - (n-1)/n^{\thres/4 - 1} = 1 - 1/n^{\thres/4-2} + 1/n^{\thres/4-1}$. Further union bounding with the $\geq 1 - 1/n^{\thres/4-1}$ probability that the number of active voters never exceeds $2 \thres \OPT_i \ln n$, as given by \Cref{lem:mng-cover:v-size}, we have that with probability $\geq 1 - 1/n^{\thres/4 - 2}$, all sets added to $C$ contain at least a $\frac{1}{5\OPT_i}$ fraction of $U_{\overline{C}}$.
\end{proof}

\begin{lemma}\label{lem:mng-cover:c-size}
    With probability $\geq 1 - 1/n^{\thres/4 - 2}$, $|C|$ never exceeds $5\OPT_i \ln n$.
\end{lemma}

\begin{proof}
    Let $C_0 = \varnothing$ and let $C_i$ denote the state of $C$ after adding the $i^\text{th}$ set. By \Cref{clm:mng-cover:high-coverage}, with probability $\geq 1 - 1/n^{\thres/4 - 2}$, $|U_{\overline{C_{k+1}}}| \leq (1 - \frac{1}{5 \OPT_i})|U_{\overline{C_k}}|$ for all $k$. Under this condition, if we let $t = 5 \OPT_i \ln n$, then
    \begin{align*}
        |U_{\overline{C_t}}| & \leq |U_{\overline{C_0}}|\left(1 - \frac{1}{5\OPT_i}\right)^t
     \leq (n-1) \left(1 - \frac{1}{5\OPT_i}\right)^{5\OPT_i \ln n}
        \leq (n-1)e^{-\ln n} < 1.
    \end{align*}
    So, the algorithm terminates having added at most $5\OPT_i \ln n$ sets to $C$.
\end{proof}

% \begin{lemma}\label{lem:mng-cover:sufficient-limit}
%     If $\ell \geq \MaxLimit$ or $\ell \geq n$, then with probability $\geq 1 - \frac{1}{n^3}$, \Cref{alg:ng_cover} returns a set $Q \subseteq P$ of size $O(\OPT{}_i \log n)$ such that for any $p_k \in U$ there is some $p_c \in Q$ satisfying $p_k \in S_{i \to c}$.
% \end{lemma}

\noindent We are now ready to prove our main result on approximation:
\begin{proof}[Proof of \Cref{thm:mng-cover:logn-approx}]
    By \Cref{lem:mng-cover:correctness}, \NGCover{} (\Cref{alg:ng_cover}) always returns either a cover of $U$ or \FAIL.
    By \Cref{lem:mng-cover:v-size,lem:mng-cover:c-size}, $|V|$ remains $\leq 2\thres\OPT_i\ln n$ and $|C|$ remains $\leq 5\OPT_i\ln n$ with probability $\geq 1 - 1/n^{\thres/4 - 2}$. Therefore, if $\ell \geq \min(n,(2\thres + 5)\OPT_i\ln n)$, with probability $\geq 1 - 1/n^{\thres/4 - 2}$, neither failure condition on Lines~\ref{alg:mng-cover:v-limit-reached} and \ref{alg:mng-cover:c-limit-reached} is ever satisfied. Thus, \NGCover{} returns a cover of $U$, and the returned cover, $C \cup V$, has size $\leq \min(n,|C| + |V|) \leq \min(n, (2\thres + 5)\OPT_i\ln n) \leq \ell$. 
\end{proof}

\subsubsection{Running Time}\label{sec:votecover-runtime}
With \Cref{thm:mng-cover:logn-approx} in place, we next analyze the running time of \Cref{alg:ng_cover}.

\begin{lemma}\label{thm:mng-cover:n-opt-time}
\NGCover{}$(P, p_i, U, \ell,\gamma)$ runs in $\tilde O((f + |U|) \min(\ell, \thres\OPT_i))$ time with probability $\geq 1 - 1/n^{\thres / 4 - 2}$, where $f$ is the maximum number of sets $S_{i \to j}$ that any $p_k \in U$ is contained in.
\end{lemma}

\begin{proof}
Any work outside of the main while-loop trivially takes $O(n)$ time. Lines~\ref{alg:mng-cover:sample-voters-start}-\ref{alg:mng-cover:sample-voters-end} also take $O(n)$ time total, because each iteration removes one element from $U$, which is initialized with size at most $|U| \leq n$ elements. So, we focus on analyzing Lines~\ref{alg:mng-cover:cast-votes-loop}-\ref{alg:mng-cover:enough-votes-end}. We divide the cost into four parts.
\smallskip
% \begin{enumerate}
%     \item Line~\ref{alg:mng-cover:cast-votes}: Casting votes, but specifically by voters that are eventually removed from $V$.
%     \item Line~\ref{alg:mng-cover:cast-votes}: Casting votes, but specifically by voters that remain in $V$ when the algorithm terminates.
%     \item Lines~\ref{alg:mng-cover:enough-votes-start}-\ref{alg:mng-cover:enough-votes:remove-voters-end}: Adding sets to $C$ and removing votes for voters that are covered.
%     \item Lines~\ref{alg:mng-cover:enough-votes:remove-nominees-start}-\ref{alg:mng-cover:enough-votes:remove-nominees-end}: Removing covered points from $U$ after a set is added to $C$.
% \end{enumerate}

\noindent\textit{Cost of casting votes, specifically by voters that are removed from $V$ (Line~\ref{alg:mng-cover:cast-votes}).} 
% Voters are only removed from $V$ on Line~\ref{alg:mng-cover:enough-votes:remove-nominees-end}, and only those contained in $voters[j]$ for some $p_j$ added to $C$. 
By Line~\ref{alg:mng-cover:enough-votes-cond}, the number of voters removed from $V$ for a given $p_j$ is exactly $\thres\ln n$. By \Cref{lem:mng-cover:c-size}, the number of sets added to $C$ is at most $5\OPT{}_i \ln n$ with probability $\geq 1 - 1/n^{\thres / 4 - 2}$ if $\ell$ is sufficiently large, or $\ell$ otherwise. Thus, the total number of voters removed from $V$ is at most $O(\min(\ell, \thres\OPT{}_i \log n) \log n)$. Each such voter casts at most $f$ votes, which requires $f$ $\SetOf$ queries, each running in $O(1)$ time. So the total cost of recording votes by voters that are at some point removed from $V$ is $O(f \cdot \min(\ell, \thres\OPT{}_i \log n) \log n)$.\\
\\
\noindent\textit{Cost of casting votes by voters that remain in $V$ (Line~\ref{alg:mng-cover:cast-votes}).} By \Cref{lem:mng-cover:v-size}, the number of voters in $V$ at any time is at most $O(\thres\OPT{}_i \log n)$ if $\ell$ is sufficiently large, or $\ell$ otherwise. Each such voter cast at most $f$ votes, so the total number of votes cast by voters that are never removed from $V$ is $O(f\cdot \min(\ell, \thres\OPT{}_i \log n))$.\\
\\
\noindent\textit{Cost of undoing votes (Lines~\ref{alg:mng-cover:enough-votes-start}-\ref{alg:mng-cover:enough-votes:remove-voters-end}).} When a voter is removed from $V$, the cost of removing its votes is identical to the initial cost of adding them in the first place, which is bounded by  $O(f \cdot \min(\ell, \thres\OPT{}_i \log n) \log n)$.\\
\\
\noindent\textit{Cost of removing covered points from $U$ (Lines~\ref{alg:mng-cover:enough-votes:remove-nominees-start}-\ref{alg:mng-cover:enough-votes:remove-nominees-end}).}  Let $m$ be the initial size of the input $U$.
Each time some $p_j$ is added to $C$, we do a full scan over $U$ and use membership queries to remove any points covered by $S_{i \to j}$. If the algorithm hits the limit $\ell$ of sets added, it is easy to see that this takes at most $O(m\ell)$ time. If not, recall from \Cref{clm:mng-cover:high-coverage} that each set added to $C$ covers at least a $\frac{1}{5\OPT{}_i}$ fraction of $U_{\overline{C}} = U \cup V$ with high probability. It follows that after adding the $t^\text{th}$ set, the cost of the corresponding scan over $U$ is $O(m \cdot (1 - \frac{1}{5\OPT{}_i})^t)$. Thus, the cost of removing covered points from $U$ is bounded by 
\begin{align*}
    O\left(\sum_{t = 0}^{5\OPT{}_i \ln n}m\left(1 - \frac{1}{5\OPT{}_i}\right)^t\right) = O\left(\sum_{t = 0}^\infty m\left(1 - \frac{1}{5\OPT{}_i}\right)^t\right) = O\left(m\OPT{}_i\right).
\end{align*}
So, the total cost of removing covered points is at most $O(m\cdot \min(\ell, \OPT{}_i))$.\\
\\
Combining all costs, we conclude that Algorithm~\ref{alg:mng-cover} runs in time $\tilde O((f + |U|) \min(\ell, \thres\OPT_i))$.
\end{proof}
\noindent \Cref{thm:ng_cover_analysis} follows as an immediate consequence of \Cref{thm:mng-cover:n-opt-time} and \Cref{thm:mng-cover:logn-approx}.

\subsubsection{Graph Construction in $\tilde O(n^{2.5})$}

On its own, \NGCover{} already implies a simple $\tilde O(n^{2.5})$ time algorithm for constructing a near optimally sparse navigable graph. This initial result improves on a naive $O(n^3)$ time implementation of greedy set cover and motivates our stronger $\tilde{O}(n^2)$ time result from \Cref{thm:main_navigability}.

\begin{theorem}\label{thm:mng-simple}
    There is an algorithm that, given a set of points $P$ and a distance function $d$, runs in $\tilde O(n \cdot \OPT{}^*) \leq \tilde O(n^{2.5})$ time and returns a navigable graph $G$ on $P$ with probability $\geq 1 - \frac{1}{n^c}$ for any constant $c \geq 1$, in which each node $p_i \in P$ has $O(\OPT{}_i \log n)$ out-neighbors.
\end{theorem}
Recall that $\OPT{}^* = \sum_{i=1}^n \OPT{}_i$ and we always have $\OPT{}^* \leq 2n^{1.5}$ by \Cref{thm:opt-upper-bound}. So, while we obtain a worst case bound of $\tilde{O}(n^{2.5})$, the result is output sensitive, improving if $G$ has a sparser navigable graph.
\Cref{thm:mng-simple} is obtained by simply running \NGCover{} independently for each $p_i$, as show in Algorithm~\ref{alg:mng-simple}.

\begin{algorithm}[htb]
    \caption{SimpleSparsestNavigableGraph($P, d$)}
    \label{alg:mng-simple}
    \begin{algorithmic}[1]
        \Require A point set $P$ of size $n$ and a distance function $d$. A constant $c \geq 1$ that controls failure probability.
        \Ensure A navigable graph on $P$.
        \algrule
        \State $G \gets (P, \varnothing)$
        \State Preprocess $P$ as in Claim~\ref{clm:sublinear_access}
        \ForEach{$p_i \in P$}
            \State $N_i \gets \NGCover{}\left(P, p_i, P \setminus \{p_i\}, n, 4c + 12\right)$
        \EndFor
        \State \Return $G$
    \end{algorithmic}
\end{algorithm}

\begin{proof}[Proof of \Cref{thm:mng-simple}]
     By \Cref{thm:mng-cover:logn-approx}, we can invoke \NGCover{} on each $p_i$ to successfully cover $P \setminus \{p_i\}$ with at most $O(\OPT{}_i \cdot \log n)$ neighbors simply by setting $\ell$ to $n$. All calls to \NGCover{} simultaneously succeed with probability at least $1 - 1/n^{\gamma/4 - 3}$ where $\gamma = 4c + 12$, which is at least $1 - 1/n^c$ for $n \geq 2$.

    Because there are $n$ sets $S_{i \to j}$ for each $i$, the frequency of any point in any set cover instance is trivially at most $n$. Moreover, the size of $|U| = P \setminus \{p_i\}$ is $n - 1$ for any $i$. Thus, by \Cref{thm:mng-cover:n-opt-time}, the call to \NGCover{} on point $p_i$ runs in $\tilde O((n + n) \min(\ell, c\OPT{}_i)) = \tilde O(n \OPT{}_i)$ time. Therefore, Algorithm~\ref{alg:mng-simple} runs in $\tilde O(\sum_{p_i \in P} n \OPT{}_i) = \tilde O(n \OPT{}^*)$ time. 
    % This is output sensitive, that is, it runs faster when $\OPT{}^*$ is smaller, but recall from \Cref{thm:opt-upper-bound} that any set of points always admits a navigable graph with $\leq n^{3/2}$ edges. Therefore, we can bound the running time for Algorithm~\ref{alg:mng-simple} by $\tilde O(n^{2.5})$ in general.
\end{proof}

% A nice qualitative property of Algorithm~\ref{alg:mng-simple} is that it makes no blanket assumptions about the points covered by arbitrary sets other than those that are trivially true. 
As explained in Section~\ref{sec:generic_applications}, a nice property of Algorithm~\ref{alg:mng-simple} is that it readily applies to related graph construction problems, such as constructing near optimally sparse $\alpha$-shortcut reachability and $\tau$-monotonic graphs. We only require  minor modifications to how we preprocess $P$ to support set access queries.

\subsection{Sparse Navigable Graph Construction in $\tilde O(n^2)$}\label{sec:quadratic-mng}
To obtain the improved $\tilde O(n^2)$ running time of our main result, \Cref{thm:main_navigability}, we make two major additions to the algorithm discussed above. These additions are summarized in \Cref{subsec:preprocess_overview}. For each set cover instance $\mathcal{I}_i$, we choose a value $\beta_i$ and prepend the call of \NGCover{} on $\mathcal{I}_i$ with the following two procedures:
\begin{itemize}
    \item We add $\beta_i$ random edges from $p_i$ to other points in $P$. This covers any points in instance $\mathcal{I}_i$ that are contained in many sets, which would otherwise cast many votes in \NGCover{}.
    \item We assign $p_i$ to an arbitrary group $K \subseteq P$ with size $\beta_i$ and add a bi-directional clique on $K$. Doing so bounds the number of uncovered points across all set cover instances corresponding to nodes in $K$.
\end{itemize}
After these preprocessing steps, we remove all points covered by the edges added, and then run \NGCover{} on the remaining points. We discuss how preprocessing impacts the running time of \NGCover{} in Section~\ref{subsec:mng-full:prelims}. Then, in Section~\ref{subsec:mng-full:algorithm}, we explain how to select $\beta_i$ for each $\mathcal{I}_i$ in order to obtain an $O(\log n)$-approximate minimum navigable graph in $\tilde O(n^2)$ time.

\subsubsection{Random Edges and Cliques}\label{subsec:mng-full:prelims}
In \Cref{clm:random_edges} we showed that, if we add $\beta$ random out-neighbors to $p_i$, then any point contained in $\geq cn\ln(n)/\beta$ sets in set cover instances $\mathcal{I}_i$ is covered with probability $> 1 - 1/n^c$. We leverage this fact to allow \NGCover{} to run faster.

\begin{claim}\label{clm:mng-full:random-edges-cover}
    If $R \subseteq P$ is a uniform random sample of size $\beta$, then running \NGCover{} on point $p_i$ with size limit $\ell$, sufficiently large constant $\thres$, and input set $U_i = P \setminus \{p_i\} \setminus \bigcup_{p_r \in R}S_{i \to r}$, takes $\tilde O((cn/\beta + |U_i|) \ell)$ time with probability $\geq 1 - 1 / n^c$ for any constant $c \geq 1$.
\end{claim}

\begin{proof}
    By \Cref{clm:random_edges} and a union bound, the maximum frequency, $f$, of any point in $U_i$ is $(c+1)n\log n/\beta$ with probability $\geq 1 - 1 / n^c$ for any $c \geq 1$.
    We then simply invoke \Cref{thm:mng-cover:n-opt-time}, which states that \NGCover{} will take time $O(f\cdot \min(\ell, \thres\OPT_i\log n)\log n + |U_i|\min(\ell, \thres\OPT_i)) \leq O((f\log n + |U_i|)\ell) \leq O((cn\log^2(n)/\beta + |U_i|)\ell)$.
    % , which shows that \NGCover{} runs in time $\tilde O((f + |U|) \min(\ell, \thres\OPT_i))$.
    % \begin{align*}
    %     O(f\min(\ell, \OPT{}_i \log n)\log n + |U| \min(\ell, \OPT{}_i)) & \leq O((n / \beta)\min(\ell, \OPT{}_i \log n)\log^2 n + |U| \min(\ell, \OPT{}_i)).
    %     % \\& = \tilde O((n/\beta + |U|)\min(\ell, \OPT{}_i))
    % \end{align*}
\end{proof}

 We first note that this bound already suffices to show that adding roughly $(\OPT{}^*/n)\log n$ random neighbors to each point is enough to construct a graph with $O(\OPT{}^* \log n)$ edges in $\tilde O(n^{2.25})$ time. Obtaining an $\tilde O(n^2)$ bound, however, requires more work. In particular, 
the running time bound in \Cref{clm:mng-full:random-edges-cover} contains two terms. The first term, involving $(n/\beta)\min(\ell, \OPT{}_i)$ can be bounded by $n$ for each $i$ if we add $\beta \approx \OPT{}_i$ random edges to each point, which would still lead to a near optimally sparse graph.  While $\OPT{}_i$ is not known in advance, it can be identified easily via an exponential search. 

The second term involving $|U_i|\min(\ell, \OPT{}_i)$ is more difficult.
In particular, suppose
% the first term $(n/\beta)\min(\ell, \OPT{}_i)$ is already sufficient if we wish to achieve the desired $\tilde O(n^2)$ time graph construction, as we will show later. That is, the only bottleneck remaining is the $|U|\min(\ell, \OPT{}_i)$ term, coming from the cost of eliminating covered points from $U$ upon adding a set to the cover. Let us briefly illustrate the issue with this term and why it prevents our algorithm as described so far from running in $\tilde O(n^2)$ time.
that $\ell$ is set high enough that \NGCover{} returns a cover of $U_i$ (see \Cref{thm:ng_cover_analysis}). Then this term becomes $\tilde{O}(|U_i|\OPT{}_i)$. Suppose we add $\beta = \OPT^*/n$ edges to each node, which is small enough that we can still construct a graph with ${O}(\OPT^*)$ edges. As shown in \Cref{cor:low_total_points}, after doing so, we will have $\sum_{i=1}^n |U_i| = \tilde{O}(n^2/\beta)$. Suppose each $|U_i|$ has roughly the same size of $\tilde{O}(n/\beta)$. Then we would obtain a total running time of roughly:
\begin{align*}
    \tilde{O}\left(\sum_{i=1}^n |U_i|\OPT{}_i\right) = \tilde{O}\left(\sum_{i=1}^n \frac{n}{\beta} \OPT{}_i\right) = \frac{n^2}{\OPT^*}\tilde{O}\left(\sum_{i=1}^n \OPT{}_i\right) = \tilde{O}(n^2).
\end{align*}
Unfortunately, however, we do not expect $U_1, \ldots, U_n$ to be uniformly sized in general after random edges are added. Moreover, while a similar argument could be used to obtain $\tilde{O}(n^2)$ runtime if $\OPT_i \approx \OPT^*/n$ for all $i$, again this will not hold in general. Different nodes could have widely varying optimal degrees.

% contributing $O(\sum_{p_i \in P}|P \setminus \{p_i\} \setminus \bigcup_{p_r \in R_i}S_{i \to r}| \cdot \OPT{}_i)$ total time to the full graph construction algorithm for randomly-chosen sets $R_1, \dots, R_n$. If the random sets are able to eliminate a similar number of points per instance, or if $\OPT{}_i$ is similar across all instances, then \Cref{cor:low_total_points} implies an $\tilde O(n^2)$ running time. However, this may not always be the case, depending on the point set. For a particular $i$, if $R_i$ is unable to cover a large number of points in instance $i$ (because they might all be highly selective, that is, each one is contained in a small number of sets), and $\OPT{}_i$ is large, then the product $|P \setminus \{p_i\} \setminus \bigcup_{p_r \in R_i}S_{i \to r}| \cdot \OPT{}_i$ will be large. For example, suppose almost all (i.e. $\Omega(n)$) points remain uncovered in instance $i$ after adding random edges. Suppose also that it requires a large number of sets (e.g. $\Omega(\sqrt{n})$) to cover those uncovered points. Then it would take $\Omega(n^{3/2})$ time to solve such a set cover instance. It can in fact be shown that there exist simple point sets that with high probability induce $\Omega(n^{5/8})$ such set cover instances, resulting in $\Omega(n^{17/8})$ time navigable graph construction.

Our second preprocessing technique allows us to circumvent these issues by adding edges in a dependent way (via cliques), which allow us to reduce the number of uncovered points in a more targeted way across different set cover instances. Before getting into details, we first consider how the runtime of \NGCover{} is impacted by the addition of a \emph{single} clique on a subset of nodes $K$ (in addition to the random edges)

\begin{claim}\label{clm:mng-full:random-edges-and-cliques}
    If $R_1, \dots, R_n \subset P$ are uniform random samples of size $\beta$ and $K \subseteq P$ has size $\leq \beta$, then running \NGCover{} on all $p_i \in K$ with size limit $\ell$, sufficiently large constant $\gamma$, and input sets $U_i = P \setminus \{p_i\} \setminus \bigcup_{p_r \in R_i} S_{i \to r} \setminus \bigcup_{p_k \in K} S_{i \to k}$ takes $\tilde O(cn \ell)$ total time with prob. $\geq 1 - 1/n^{c+1}$ for any constant $c \geq 1$.
\end{claim}

\begin{proof}
    By \Cref{clm:mng-full:random-edges-cover}, running \NGCover{} on a single instance $p_i$ to cover $P \setminus \{p_i\} \setminus \bigcup_{p_r \in R_i} S_{i \to r}$ (and by extension any subset $U_i$ thereof) takes $O((cn\log^2(n) / \beta + |U_i|)\ell)$ time with probability $\geq 1 - 1/n^{c+2}$ for any $c \geq 1$. Taking a union bound over all instances in $K$, this holds simultaneously across all $\leq \beta$ of those instances with probability $\geq 1 - \beta/n^{c+2} \geq 1 - 1/n^{c+1}$. Recall from \Cref{lem:cliques} that $\sum_{p_i \in K}|U_i| \leq n - \beta \leq n$. Therefore, the total time for running \NGCover{} on all instances in $K$ is:
    \[
        O\left(\sum_{p_i \in K} \left(\frac{cn\log^2 n}{\beta} + |U_i|\right)\ell\right) = O\left(\ell\sum_{p_i \in K}\frac{cn\log^2 n}{\beta} + \ell \sum_{p_i \in K} |U_i|\right) = O\left(cn\ell\log^2 n + n\ell\right) = \tilde O(cn\ell).\qedhere
    \]
\end{proof}

\subsubsection{Full Algorithm}\label{subsec:mng-full:algorithm}
With the above claims in place, we are equipped to augment Algorithm~\ref{alg:mng-simple} with both random edges and arbitrary clique additions. The pseudocode of our final method that does so is given in Algorithm~\ref{alg:mng-full}.
We first give a high-level overview of the algorithm. The main loop, located on Line~\ref{alg:mng-full:main-loop}, iterates over doubling values of $\ell$. This performs exponential search on $\OPT{}_i$ for all instances $\mathcal{I}_i$ in parallel. We know that \NGCover{} with size limit $\ell$ will successfully return a cover for a point $p_i$ before $\ell$ exceeds $\OPT{}_i \ln n$ by more than a constant factor. Thus, we can supplement the cover with $O(\ell)$ other neighbors without violating our promised $O(\log n)$ approximation factor. Under this assumption, we add $\ell$ random edges and $\ell$ clique edges to each $p_i$. We then efficiently filter out any points covered by these edges as described in Section~\ref{subsec:mng-full:prelims}, then run \NGCover{} with limit $\ell$ on each instance to cover the remaining points. For any instances that successfully return covers, we can build a full neighbor list by combining the random edges, the clique edges, and the cover from \NGCover{}. We subsequently remove these instances from contention, i.e., they will not be considered in future iterations. For any instances where \NGCover{} fails, we conclude that $\OPT_i$ must be much larger than the current choice of $\ell$ and simply try again in the next iteration.

\begin{algorithm}[htb]
\caption{SparsestNavigableGraph($P, d$)}
\label{alg:mng-full}
\begin{algorithmic}[1]
\Require A point set $P$ of size $n$, a distance function $d$, and constants $c, \hat{c} > 0$ that control failure probability.
\Ensure A navigable graph on $P$ representing by out-neighborhoods $N_1, \ldots, N_n \subseteq P$ for each $p_i \in P$.
\State Preprocess $P$ as in \Cref{clm:sublinear_access}. Initialize $Q \gets P$.
\For{$\ell \gets 2^1, 2^2, \dots, 2^{\lceil \lg n \rceil}$}\label{alg:mng-full:main-loop}\Comment{Exponential search on $\OPT{}_i$}
    \State $U_{i_1}, U_{i_2}, \dots, U_{i_{|Q|}} \gets \varnothing$ for each $p_{i_j} \in Q$ \label{alg:mng-full:initialize-uncovered}
    \State $R_{i_1}, R_{i_2}, \dots, R_{i_{|Q|}} \gets$ samples of $P$ with size $(c+3)\ell$ for each $p_{i_j} \in Q$\Comment{Choose random edges}
    \State Arbitrarily partition $Q$ into parts $K_1, K_2, \dots, K_{\lceil |Q| / \ell \rceil}$ of size $\leq \ell$\label{alg:mng-full:initialize-cliques}\Comment{Choose arbitrary cliques}
    \For{$K = K_1, K_2, \dots, K_{\lceil |Q| / \ell \rceil}$}\label{alg:mng-full:cliques-loop}
        \ForEach{$p_k \in P \setminus K$}\label{alg:mng-full:filter-covered-start}\Comment{Identify points not yet covered by random edges or clique edges}
            \State $p_i \gets \argmin_{p \in K}d(p, p_k)$\label{alg:mng-full:argmin-clique}
            % \If{(the number of sets $S_{i \to k}$ containing $p_k) \leq \HF(n / \ell) \log n$}\label{alg:mng-full:filter-high-freq}
                \If{$p_k \not\in S_{i \to r}$ for all $p_r \in R_i$}\label{alg:mng-full:filter-random-edges}
                    \State $U_i \gets U_i \cup \{p_k\}$\label{alg:mng-full:filter-covered-end}
                \EndIf
            % \EndIf
        \EndFor
        \ForEach{$p_i \in K$}\label{alg:mng-full:run-covers-start}\Comment{Run \NGCover{} to cover all identified uncovered points}
        % need to check 4c + 12 is actually correct
            \State $C \gets \NGCover{}(P, p_i, U_i, \ell,4\hat{c} + 12)$ \label{line:votecover_call}
            \If{$C \not=$ \textsc{Fail}}
                \State $Q \gets Q \setminus \{p_i\}$\label{alg:mng-full:mark-instance-completed}
                \State $N_i \gets R_i \cup K \cup C \setminus \{p_i\}$\label{alg:mng-full:run-covers-end}
            \EndIf
        \EndFor
    \EndFor
\EndFor
\State \Return Out-neighborhoods $N_1, \ldots, N_n$ representing directed graph $G$.
\end{algorithmic}
\end{algorithm}

% We now analyze the approximation factor and running time of Algorithm~\ref{alg:mng-full}.

\subsubsection{Correctness and Approximation Factor}\label{subsec:mng-full:approx}

Here we will prove that Algorithm~\ref{alg:mng-full} always returns a navigable graph where each node's out-degree is optimal to within an $O(\log n)$ factor. This proof follows easily from results proven earlier in this paper.

\begin{claim}\label{obs:mng-full:correctness}
    Algorithm~\ref{alg:mng-full} returns a navigable graph.
\end{claim}

\begin{proof}
    Since the algorithm keeps increasing $\ell$ to be $\geq n$ if the \NGCover{} call on \Cref{line:votecover_call} fails, it is not hard to see that \NGCover{} eventually returns a solution for all $p_i$ (indeed, the algorithm succeeds with probability $1$ if $\ell \geq n$). So, it suffices to show that the returned solution, $R_i \cup K \cup C \setminus \{p_i\}$ is a proper cover for set cover instance $\mathcal{I}_i$. This is true because $C$ is a cover for $U_i$ and a point $p_k$ is only not added to $U_i$ if it is already covered by an edge in $K$ or in $R_i$. In particular, if $p_i \neq \argmin_{p \in K}d(p, p_k)$ then it is covered by the edge from $p_i$ to $\argmin_{p \in K}d(p, p_k)$ in the clique. 
\end{proof}

% \begin{claim}\label{clm:mng-full:returned-size}
%     Within a given iteration of the $\ell^*$-loop, if the algorithm returns a graph, then that graph has at most $4 \ell^* \ln n$ edges.
% \end{claim}

% \begin{proof}
%     It's clear to see that $\lambda$ tracks the sum of the sizes of all sets returned by \NGCover{} for the current choice of $\ell^*$. Moreover, the restart condition on Line~\ref{alg:mng-full:restart-condition} jumps to the next iteration of the $\ell^*$-loop if $\lambda + |Q| \cdot \ell > \ell^* \ln n$ for a given iteration of the $\ell$-loop. In that $\ell$-loop iteration, all calls to \NGCover{} are made with size limit $\ell$. Thus, each cover $C$ returned has size at most $\ell$, meaning $\lambda$ can only increase by at most $|Q| \cdot \ell$ in this iteration. It follows that $\lambda$ never exceeds $\ell^* \ln n$. Thus, if the algorithm ever returns a graph, then the sum size of all $C$'s returned by \NGCover{}, given by $\lambda$, is at most $\ell^* \ln n$. For that choice of $\lambda$, all random samples $R_i$ have size $2(\ell^* / n) \ln n$, and all random cliques $K$ have size $(\ell^* / n) \ln n$. Thus, the sum sizes of all $R_i \cup K$ added to each $N_G(p_i)$ is at most $n \cdot (3 (\ell^* / n) \ln n) \leq 3 \ell^* \ln n$. Therefore, the total number of edges added to $G$ is at most $4 \ell^* \ln n$.
% \end{proof}

\begin{claim}\label{clm:mng-full:large-instances-remain}
    With probability $\geq 1 - 1/n^{\hat{c}}$, Algorithm~\ref{alg:mng-full} maintains the invariant that $Q$ only contains points $p_i$ where $\OPT_i > \ell/((16\hat{c} + 58)\ln n)$.
\end{claim}

\begin{proof}
    From \Cref{thm:mng-cover:logn-approx}, when run with  $\ell \geq (8\hat{c} + 29) \OPT_i \ln n$, \NGCover{} successfully returns a cover of size $\leq \ell$ with probability $\geq 1 - 1/n^{\hat{c} + 1}$. Since $\ell$ doubles with each iteration of the main loop, it follows that for each $i$, there is an iteration prior to $\ell$ exceeding $(16\hat{c} + 58)\OPT_i \ln n$ where $(8\hat{c} + 29) \OPT{}_i \ln n \leq \ell < (16\hat{c} + 58) \OPT_i \ln n$. By a union bound, it holds simultaneously over all $i$ that \NGCover{} returns a cover for $\mathcal{I}_i$ on that corresponding iteration with probability $\geq 1 - 1 / n^{\hat{c}}$, and possibly even earlier. Subsequently, the corresponding $p_i$ are removed from $Q$ by Line~\ref{alg:mng-full:mark-instance-completed}. Thus, with this probability, all $p_i$ remaining in $Q$ satisfy $\ell < (16\hat{c} + 58) \OPT_i \ln n$, or equivalently $\OPT_i > \ell/((16\hat{c} + 58) \ln n)$.
\end{proof}

\begin{corollary}
\label{cor:final_approx}
    With probability $\geq 1 - 1/n^{\hat{c}}$, the graph returned by Algorithm~\ref{alg:mng-full} has $O(c\hat{c}\OPT{}^* \log n)$ edges, and each $p_i$ has $O(c\hat{c}\OPT{}_i \log n)$ out-neighbors.
\end{corollary}

\begin{proof}
    Each point's set of neighbors is constructed as soon as \NGCover{} successfully returns some set $C$ for that point. By \Cref{clm:mng-full:large-instances-remain}, with probability $\geq 1 - 1 / n^{\hat{c}}$, this happens for $p_i$ when $\ell \leq (16\hat{c} + 58) \OPT{}_i \ln n$. On the corresponding iteration of the main loop, $R_i$ has size $(c + 3)\ell$, and all partitions $K$ have size $\leq \ell$. Therefore,  each $p_i$ receives at most $|R \cup K \cup C| \leq (c + 5)\ell \leq (c + 5)(16\hat{c} + 58) \OPT{}_i \ln n = O(c\hat{c}\OPT_i \log n)$ out-neighbors.
    By extension, the total number of edges in the graph is $O(c\hat{c}\OPT{}^* \log n)$. 
\end{proof}

% As a final footnote to this correctness analysis, we point out that although we present fixed probability bounds for our algorithm, they are easily tunable, with the sources of the probabilities stemming from \Cref{clm:mng-cover:high-coverage}, \Cref{lem:mng-cover:v-size}, and \Cref{obs:mng-full:correctness}. These can be adjusted by changing the constants on Line~\ref{alg:mng-cover:enough-votes-cond} of Algorithm~\ref{alg:mng-cover} and Line~\ref{alg:mng-full:filter-high-freq} of Algorithm~\ref{alg:mng-full}, as well as that in \Cref{clm:mng-cover:high-coverage} bounding the number of points covered by highly-voted sets. This trades off with constants from the algorithm's approximation factor and running time.

\subsubsection{Running Time}

To finish the proof of \Cref{thm:main_navigability}, we analyze the running time for Algorithm~\ref{alg:mng-full}.

\begin{lemma}
    Each iteration of the main loop takes $\tilde O(cn^2)$ time with probability $\geq 1 - 1 / n^{c + 1}$.
\end{lemma}

\begin{proof}
    Lines~\ref{alg:mng-full:cliques-loop}-\ref{alg:mng-full:run-covers-end} constitute the nontrivial running time of a given iteration of the main loop. For each clique $K$, these lines perform two main steps.
    \begin{enumerate}
        \item Lines~\ref{alg:mng-full:filter-covered-start}--\ref{alg:mng-full:filter-covered-end}: The algorithm initializes the set cover instances within $K$. This involves filtering out all points in these instances which are already covered either by clique edges or random edges.
        \item Lines~\ref{alg:mng-full:run-covers-start}--\ref{alg:mng-full:run-covers-end}: The algorithm runs \NGCover{} on each set cover instance in $K$ to cover all remaining points that Lines~\ref{alg:mng-full:filter-covered-start}--\ref{alg:mng-full:filter-covered-end} identified as not yet covered.
    \end{enumerate}

    \noindent\textit{Initializing set cover instances.} Line~\ref{alg:mng-full:argmin-clique} is done using a single scan over $K$, computing the distance between $p_k$ and each member of $K$ to determine the minimum among them. This takes $O(|K|) = O(\ell)$ time per execution of Line~\ref{alg:mng-full:argmin-clique}. Likewise, Line~\ref{alg:mng-full:filter-random-edges} performs a scan over $R_i$, which takes $O(|R_i|) = O(\ell)$ time. Lines~\ref{alg:mng-full:argmin-clique} and \ref{alg:mng-full:filter-random-edges} are run $n - |K| = O(n)$ times per clique, of which there are at most $O(n / \ell)$. Thus, the total amount of work across all cliques is $O(\ell) \cdot O(n) \cdot O(n / \ell) = O(n^2)$.\\
    \\
    \noindent\textit{Running \NGCover{}.} By \Cref{clm:mng-full:random-edges-and-cliques}, the calls to \NGCover{} for a given clique $K$ run in combined time $O(cn\ell)$ with probability $\geq 1 - 1/n^{c+2}$. Since there are at most $\lceil n / \ell\rceil$ cliques, this takes $O(cn\ell\log^2 n) \cdot O(n / \ell) = O(cn^2\log^2 n)$ time in total with probability $\geq 1 - 1/n^{c+1}$ (by a union bound).
\end{proof}

% \begin{observation}
%     The main loop on Line~\ref{alg:mng-full:main-loop} runs for $O(\log n)$ iterations.
% \end{observation}

% \begin{proof}
%     By \Cref{thm:mng-cover:logn-approx}, \NGCover{} always successfully returns a set when the provided limit $\ell$ is at least $n$. It only takes $\lceil \log_2 n \rceil$ iterations for $\ell$ to grow to $n$, so all instances receive a valid neighbor list within that many iterations.
% \end{proof}

\begin{corollary}\label{cor:final_runtime}
    Algorithm~\ref{alg:mng-full} runs in $\tilde O(cn^2)$ time with probability $\geq 1 - 1/n^c$.
\end{corollary}

\begin{proof}
    This follows easily from the previous lemma. The main loop runs for at most $\lceil \log_2 n \rceil$ iterations, each taking $O(n^2 \log^2 n)$ time with probability $\leq 1 - 1/n^{c + 1}$. Taking a union bound over $\leq \lceil \log_2 n \rceil$ iterations, this holds for all levels simultaneously with probability $\geq 1 - \lceil \log_2 n \rceil/n^{c + 1}> 1 - 1/n^c$, in which case the algorithm takes a total of $O(n^2 \log^3 n)$ time.
\end{proof}

\noindent\Cref{thm:main_navigability} follows as an immediate consequence of \Cref{cor:final_approx} and \Cref{cor:final_runtime}.

%% file: Hardness.tex
\section{Hardness}
In this section we establish hardness results to show that the $O(\log n)$ approximation achieved by our algorithm for constructing minimum sparsity navigable graphs is essentially tight for any polynomial-time algorithm, even when the input point set is in low-dimensional Euclidean space (e.g., $O(\log(n))$ dimensions). 
We start with preliminary terminology, definitions and lemmas we will use in proving the results.
\label{sec:lower_bounds}
\subsection{Preliminaries}
We first introduce some additional terminology related to the ``local'' definition of navigability in \Cref{def:navigability}.
\begin{definition}
Given point set $P$ and distance metric $d : P\times P\rightarrow \mathbb{R}^+$, and given a point $p\in P$ and a set $P^\prime \subseteq P$, $P^\prime \subseteq P$ is a \textbf{navigable neighborhood} for $p$ if and only if for every $p'' \in P$ where $p \neq p''$, there exists some $p'\in P^\prime$ such that $d(p', p'') < d(p, p'')$.
\end{definition}

\begin{definition}
Given point set $P$ and distance metric $d : P\times P\rightarrow \mathbb{R}^+$, and given a point $p\in P$ and a set $P^\prime \subseteq P$, we say that $P^\prime \subseteq P$ is a \textbf{minimum navigable neighborhood} for $p$ if and only if for every navigable neighborhood $P^{\prime\prime}\subseteq P$ we have $\left|P^\prime\right|\leq \left|P^{\prime\prime}\right|$.
We say that $P^\prime \subseteq P$ is a $c$-\textbf{approximate minimum navigable neighborhood} for $p$ if and only if for every navigable neighborhood $P^{\prime\prime}\subseteq P$ we have $\left|P^\prime\right|\leq c\cdot\left|P^{\prime\prime}\right|$.
\end{definition}
A graph is navigable if and only if the out-neighborhood of every node $p$ is a navigable neighborhood for $p$. Further, a graph is a minimum navigable graph if and only if the out-neighborhood of every $p$ is the minimum navigable neighborhood. Note however, that if we have a $c$-approximate minimum navigable graph in terms of \emph{total number of edges}, it does not mean that any particular neighborhood is $c$-approximate minimum. As such, proving hardness results nearly minimizing total degree (equivalently, average degree) is more challenging than proving such results for nearly minimizing the neighborhood size of any particular $p$.

\subsubsection{Defining Covering Point Sets}
Our hardness results will be based on the hardness of solving set cover, either exactly and approximately. As such, we start by describing a transformation from a set cover instance to a point set, with the property that navigable neighborhoods for the point set can be tranlated into a solution to the original set cover instance. We will instantiate this transformation in different settings to obtain various hardness results.
\begin{definition}
Given a set cover instance $(U=[n],\mathcal{F}=\{S_1,...,S_m\})$, a point set $P$ is called a \textbf{representative point set} for $(U,\mathcal{F})$ if $P_{(U,\mathcal F)}=\{p_i:i\in[n]\}\cup\{p_{S_j}:S_j\in \mathcal{F}\}\subseteq P$, meaning $P$ contains a point for every element, which we term element points, and every set, which we term set points.
Further, $P$ contains a distinguished point, $c\in P\setminus P_{(U,\mathcal F)}$ which we term the centroid.
Points in $(P\setminus P_{(U,\mathcal F)})\setminus \{c\}$, i.e the remaining points, are called solution points. We denote the set of solution points by $P_{SC}$.
\end{definition}

\begin{definition}
Given a set cover instance $(U=[n],\mathcal{F}=\{S_1,...,S_m\})$, a point set $P$ is called a \textbf{covering point set}
 for $(U,\mathcal{F})$ with respect to distance metric $d$ on $P$, if it is a representative point set, and the following properties hold:
\begin{enumerate}
    \item $\forall p\in P_{SC},p^\prime\in P\setminus\{p,c\}$, $d(p,c)<d(p,p^\prime)$. I.e., $c$ is the unique nearest neighbor for any solution point.
    \item $\forall i,j\in[m],p\in P_{SC}$, $d(p_{S_i},p_{S_{j}})<d(p_{S_i},p)$. I.e., any set point is closer to any other set point than to a solution point.
    \item $\forall p\in P_{SC},p^\prime \in P,i\in[n]$, $d(p^\prime,p_i)<d(p,p_i)$ if and only if $p^\prime=p_i$ or $p^\prime=p_{S_j}$ where $i\in S_j$. I.e., moving from a solution point to another point gets us closer to an element point $p_i$ if and only if the other point is equal to $p_i$ or is the point corresponding to a set that contains $i$.
    \item 
    $\forall p,p'\in P_{SC},p''\in (P\setminus P_{SC})\setminus \{c\}$, $d(p,p')<d(p,p'')$. I.e., any solution point is closer to any other solution point than it is to any element or set point.
\end{enumerate}
\end{definition}
\subsubsection{Properties of Covering Point Sets}
\begin{observation}
In a covering point set, an element point or a set point has a minimum navigable neighborhood which size is at most $m+n+1$.
\end{observation}
\begin{proof}
Let $p$ be an element or a set point, and consider $N=(P\setminus P_{SC})\setminus{\{p\}}$ as a candidate neighborhood for $p$. All set and element points are in $N$ and so is $c$. Further, according to property (1) of covering  point sets, $c$ is closer than $p$ to all points in $P_{SC}$, thus $N$ is a navigable neighborhood with size at most $n+m+1$. 
\end{proof}
\begin{observation}\label{lem:solutionPointsMakeSolutions}
Given a covering point set for a set cover instance $(U,\mathcal{F})$, a solution point $p\in P_{SC}$, and a navigable neighborhood $N\subseteq P$ for $p$, we can efficiently (in time polynomial in n,m) compute a set cover for $(U,\mathcal{F})$ with size at most $|N|-1$.
\end{observation}
\begin{proof}
Define the set $\hat{N}=(N\setminus P_{SC})\setminus\{c\}$. According to property (4) of covering point sets, and since $N$ is a navigable neighborhood for $p$, $\hat{N}$ is a strict subset of $N$. In particular, we would only have $\hat{N} = N$ if $N$ only contains set or element points. This cannot be the case by property (4) of covering point sets, as no set or element points are closer to $p' i\in P_{SC}$ than $p$ is. Further, according to property (3) of covering point sets, for every $i\in [n]$, either $p_i\in N$ or $p_{S_j}\in N$ where $i\in S_j$. Thus, if we swap every element point in $\hat{N}$ with a set point of an arbitrary set that contains the element, we get a collection of at most $|N|-1$ set points, whose corresponding sets cover all elements in $U=[n]$.
\end{proof}
\begin{lemma}
In a covering point set, the size of a minimum navigable neighborhood for a solution point is exactly $\OPT{}+1$ where $\OPT{}$ is the size of the minimum set cover for $(U,\mathcal{F})$.
\end{lemma}
\begin{proof}
Let $p\in P_{SC}$, if $\mathcal{F}^\prime\subseteq \mathcal {F}$ is a set cover,  then $N=\{p_{S_j}:S_j\in \mathcal{F}'\}\cup \{c\}$ i,e all points representing sets in $\mathcal{F}'$ and the centroid, are a navigable neighborhood for $p$. To see why this is true - note that $c$ is in $N$, and any set in $\mathcal{F^\prime}$ (which is also in $N)$ gets us closer to all other sets according to property (2) of covering point sets. Further, for every $i\in[n]$ there exists $S_j\in \mathcal{F^\prime}\subseteq N$ such that $i\in S_j$ and thus according to property (3) of covering point sets $d(p_{S_j},p_i)<d(p,p_i)$ - i.e we can get closer to all element points. This implies that the size of the minimum navigable neighborhood is at most $\OPT{}+1$. In the other direction - note that if $N$ is a navigable neighborhood for $p$, then according to \cref{lem:solutionPointsMakeSolutions}, there exists a set cover of size $|N|-1$, and thus $|N|\geq \OPT{}+1$. 
\end{proof}

As this lemma shows, the main strength of covering point sets, is that the problem of finding a navigable neighborhood for a solution point corresponds to computing a set cover for the original instance. A crucial part, demonstrated in the previous observations, is that the size of minimum navigable neighborhoods for set and element points does not scale with the number of solution points (due to the centroid), so by including moderately many solution points in the construction we can claim that in a navigable graph with approximate-minimum average degree, some solution point has a navigable neighborhood that reasonably approximates the minimum. This will be central to the way we use these constructions. 
\begin{corollary}\label{cor:minimumNavGraphForCoveringPointSets}
The total number of edges in a minimum navigable graph for a covering point set is at most $(n+m)(n+m+1)+m+n+|P_{SC}|+(\OPT{}+1)\cdot |P_{SC}|$.
\end{corollary}
\begin{lemma}\label{lem:solutionForCoveringPointSetsSolvesSetCover}
Let $(U,\mathcal{F})$ be a set cover instance, and let $P$ be a covering point set. Then, if $|P_{SC}|\geq(m+n)^3$, and $G$ is a navigable graph for $P$ which average degree is $r$-approximate minimum, then we can use $G$ to efficiently (in time polynomial and $m,n$), compute a $\left(r+o(1)+\Theta\left(\frac{r}{\OPT{}}\right)\right)$-approximate minimum set cover for $(U,\mathcal{F})$.
\end{lemma}
\begin{proof}
First we show there must exist some solution point in $P$ which degree in $G$ is at most $\left(r+o(1)+\frac{1}{\OPT{}}\right)\cdot (\OPT{}+1)$. Indeed let $\lambda>0$ such that all solution points have degree at least $\lambda \cdot (\OPT{}+1)$. This implies that the sum of degrees in $G$ is at least $\lambda(m+n)^3\cdot(\OPT{}+1)$. On the other hand, we know from \cref{cor:minimumNavGraphForCoveringPointSets}, and since $G$'s average degree is $r$-approximate minimum that the sum of degrees in $G$ is at most: 
\[
r\cdot\left(\left(n+m\right)\left(n+m+1\right)+m+n+\left(m+n\right)^{3}+\left(\OPT{}+1\right)\cdot\left(m+n\right)^{3}\right)
\]
So:
\[
\lambda\left(m+n\right)^{3}\cdot\left(\OPT{}+1\right)\leq r\cdot\left(\left(n+m\right)\left(n+m+1\right)+m+n+\left(m+n\right)^{3}+\left(\OPT{}+1\right)\cdot\left(m+n\right)^{3}\right)
\]
Thus:
\[
\lambda\leq\frac{r\cdot\left(\left(n+m\right)\left(n+m+1\right)+m+n+\left(m+n\right)^{3}+\left(\OPT{}+1\right)\cdot\left(m+n\right)^{3}\right)}{\left(m+n\right)^{3}\cdot\left(\OPT{}+1\right)}
\]
\[
\leq r+\frac{\Theta\left(\left(m+n\right)^{2}\right)}{\left(m+n\right)^{3}\cdot\left(\OPT{}+1\right)}+\frac{1}{\OPT{}+1}\leq r+o\left(1\right)+\frac{1}{\OPT{}}
\]
So, there exists a solution point with neighborhood in $G$ which size is at most: 
\begin{equation}\label{eq:reductionApproxRatio}
\left(r+o(1)+\frac{1}{\OPT{}}\right)\cdot \left(\OPT{}+1\right)=\OPT{}\cdot\left(r+o(1)+\Theta\left(\frac{r}{\OPT{}}\right)\right)
\end{equation}
Now, recall \cref{lem:solutionPointsMakeSolutions}, which implies that for each solution point, we can efficiently compute a set cover which size is at most the size of its neighborhood. Thus, we can use the solution point that has minimal degree, which as we've established is at most $\OPT{}\cdot\left(r+o(1)+\Theta\left(\frac{r}{\OPT{}}\right)\right)$, to compute a set cover of this size, i.e a $\left(r+o(1)+\Theta\left(\frac{r}{\OPT{}}\right)\right)$-approximate minimum set cover.
\end{proof}
\subsection{Hardness for General Distance Metrics}
The first step towards establishing hardness of approximation is, given a set cover instance $(U,\mathcal{F})$, to show how to actually construct a covering point set for this instance.
This is quite straightforward if we are allowed to use a general metric.
Specifically, we define a mapping $\varphi(U,\mathcal{F})=(P,d)$ that maps the set cover instance to a representative point set $P$ and distance metric $d$ which are specified as follows:
$P$ has $(n+m)^2$ solution points, $d$ is a symmetric function, such that $d(p,p^\prime)=0\iff p=p^\prime$. 
We specify $d$ in table form (see \Cref{fig:general_metric_point_distances}).

\begin{center}
\setlength{\fboxrule}{0.8pt} % border width (unchanged by scaling)
\setlength{\fboxsep}{8pt}    % gap between border and table

\begin{figure}[htbp]
  \centering
    \scalebox{0.9}{%  ← tweak this number
      \begin{tabular}{ c|c|c|c|c } 
  & Element points & Set points & Solution points & $c$ \\
 \hline
 Element points & 2 & $d\left(p_{i},p_{S_{j}}\right)=\begin{cases}
1 & i\in S_{j}\\
2 & else
\end{cases}$ & 2  & 2 \\
   \hline
 Set points & " & 1 & 2 & 2 \\
   \hline
 Solution points & " & " & 1.5 & 1 \\
 \hline
 $c$ & " & " & " & 0 \\
\end{tabular}
    }% end \scalebox
  \caption{Specification of distance metric $d$ on different types of points in $P$.}
\label{fig:general_metric_point_distances}
\end{figure}
\end{center}

Note that $d$ satisfies triangle inequality since the values it assigns to distances between distinct points are in $\{1,1.5,2\}$.
It is straightforward to verify that $(P,d)$ is a covering point set: For property (1) - Observe that the distance of $c$ from a solution point is smaller than all other distances in the solution points column. For property (2) observe that the distance between set points is $1$, while the distance between a set point and a solution point is $2$. For property (3), observe that the only point that has a distance smaller than $1.5$ from an element, is the element itself, and any set that contains it. For property (4) observe that the distance between solution points is $1.5$, while their distance from set and element points is $2$.
\begin{theorem}\label{thm:HardnessOfApproxGeneral}
For any $\varepsilon>0$, unless $P=NP$, there does not exist a polynomial-time algorithm that, given a set $P$ of $n$ points and a distance metric $d : P \times P \to \mathbb{R}^+$, constructs a navigable graph with average degree $\leq (\frac{1}{3}\cdot ln(n)-\varepsilon)\cdot \OPT{}$, where $\OPT{}$ is the minimum average degree of any navigable graph on $(P,d)$.
\end{theorem}
\begin{proof}
Let $\varepsilon >0$, and assume the existence of such an algorithm. Set $0<\varepsilon ^\prime=\frac{\varepsilon}{2}<\varepsilon$, and observe section $6.2$ of \cite{DinurSteurer:2014}: The instances their reduction constructs, that are $NP$-hard to $\log(n)-\varepsilon ^\prime$ approximate, have 
$N=n_{1}^{D+1}$ elements and  $M=n_{1}^{2}$ sets where $D>1$ is inversely proportional to $\varepsilon ^\prime$. So, $M\leq N$. Further, the size of the optimal solution for such an instance satisfies $\OPT{}\in\Omega\left(n_{1}\right)=\Omega\left(N^{\frac{1}{D+1}}\right)$. We show that the existence of the hypothesized algorithm implies an efficient algorithm to construct a $\ln(N)-\varepsilon^\prime$-approximate minimum set cover for such instances, implying $P=NP$.
Consider such an instance $(U,\mathcal{F})$ with $N$ elements, $M$ sets and an optimal solution of size $\OPT{}$ which all behave as we've stated above and construct a covering point set $P$ and distance metric $d$ as described in the preceding discussion - note that this point set has $N^\prime= M+N+1+(M+N)^2$ points.
We can run the hypothesized algorithm on $(P,d)$ and retrieve a navigable graph which average degree is $r$-approximate minimum where $r=\frac{1}{3}\cdot \ln(N')-\varepsilon$. According to \cref{lem:solutionForCoveringPointSetsSolvesSetCover}, we can then use this graph to efficiently compute a set cover for $(U,\mathcal{F})$ which size is at most 
\[
\OPT{}\left(r+o(1)+\Theta\left(\frac{r}{\OPT{}}\right)\right)=\OPT{}(r+o(1))
\]
\[
=\OPT{}\left(\frac{1}{3}\cdot \ln(N')-\varepsilon+o(1)\right)\leq \OPT{}\left(\frac{1}{3}\cdot \ln(N')-\varepsilon'\right)   
\]
\[
\OPT{}\left(\frac{1}{3}\cdot \ln(M+N+1+(M+N)^2)-\varepsilon'\right)
\]
\[
\leq \OPT{}\left(\frac{1}{3}\cdot \ln(N^3)-\varepsilon'\right)=\OPT{}(\ln(N)-\varepsilon')
\]
So, we have a $\ln(N)-\varepsilon '$ approximation in polynomial time for $NP$-hard-to-approximate instances, which implies $P=NP$.
\end{proof}

\subsection{Hardness for Euclidean Metrics}\label{subsec:EuclideanHardness}
To show hardness of approximation for Euclidean metrics, we demonstrate how a covering point set can be constructed in Euclidean space. From there, the proof for hardness of approximation is identical to \Cref{thm:HardnessOfApproxGeneral} and can be derived as a corollary.
Given a set cover instance $(U,\mathcal{F})$ we define a mapping $\varphi'(U,\mathcal{F})$ which maps the set cover instance to a representative point set $P\subseteq \mathbb{E}^{(m+n+2+(m+n)^{3})}$ with $(m+n)^3$ solution points, which we denote by $P_{SC}=\{\ell_1,...,\ell_{(m+n)^3}\}$, as follows:
\begin{itemize}
    \item $\forall S_{j}\in\mathcal{F}:p_{S_{j}}=\mathbf{e}_{j}+\mathbf{e}_{n+m+1}$.
    \item $\forall i\in\left[n\right]:p_{i}=n\cdot \mathbf{e}_{m+i}+\sum_{j:i\in S_{j}}\mathbf{e}_{j}$
    \item $\forall k\in\left[\left(m+n\right)^{3}\right]:\ell_{k}=\frac{1}{2}\cdot \mathbf{e}_{n+m+2}+\sqrt{\frac{3}{4}}\cdot \mathbf{e}_{n+m+2+k}$
    \item $c=\mathbf{e}_{n+m+2}$.
\end{itemize}
where $\mathbf{e}_i$ denotes the i'th standard basis vector.
Observe that the distance of $c$ or a solution point, from an element point $e_i$ is equal to $\sqrt{n^{2}+f_i+1}$ where $f_i$ is the frequency of $i$, i.e $\lvert \{S_j:i\in S_j\}\rvert$. Further, the distance of $p_{S_j}$ from $p_i$ is either $\sqrt{n^{2}+f_i}$ or $\sqrt{n^{2}+f_i+2}$ depending on whether $i\in S_j$ or $i\notin S_j$, \cref{fig:euclidean_point_distances} shows a summary of the distances between the different types of points in the representative point set we've constructed in table form.
\begin{center}
\setlength{\fboxrule}{0.8pt} % border width (unchanged by scaling)
\setlength{\fboxsep}{8pt}    % gap between border and table

\begin{figure}[htbp]
  \centering
    \scalebox{0.9}{%  ← tweak this number
      \begin{tabular}{ c|c|c|c|c } 
  & Element points & Set points & Solution points & $c$ \\
 \hline
 Element points & $\geq\sqrt{2n^{2}}$ & $d\left(p_{i},p_{S_{j}}\right)=\begin{cases}
\sqrt{n^{2}+f_i} & i\in S_{j}\\
\sqrt{n^{2}+f_i+2} & else
\end{cases}$ & $\sqrt{n^{2}+f_i+1}$  & $\sqrt{n^{2}+f_i+1}$ \\
   \hline
 Set points & " & $\sqrt{2}$ & $\sqrt{3}$ & $\sqrt{3}$ \\
   \hline
 Solution points &" & " & $\sqrt{1.5}$ & 1 \\
 \hline
 $c$ & " & " & " & 0 \\

\end{tabular}
    }% end \scalebox
  \caption{Euclidean distance between different types of points in $\varphi'(U,\mathcal{F}).$}
  \label{fig:euclidean_point_distances}
\end{figure}
\end{center}

Again, it is straightforward to verify that the resulting point set is indeed a covering point set, for example observe that if $\ell_j$ is a solution point, and $p_i$ is an element point, then $d(\ell_j,p_i)=\sqrt{n^2+f_i+1}$. Observing the table row corresponding to distances from element points we can see that the only point closer to $p_i$ than $\ell_j$ is, is a point corresponding to a set containing $i$, or $p_i$ itself, so property (3) of covering point sets is satisfied.
Overall we can derive the following theorem as a corollary:
\begin{theorem}\label{thm:HardnessOfApproxEuclidean}
For any $\varepsilon>0$, unless $P=NP$, there does not exist a polynomial-time algorithm that, given a set $P$ of $n$ points in $d$-dimensional Euclidean space, constructs a navigable graph with average degree $\leq (\frac{1}{3}\cdot ln(n)-\varepsilon)\cdot \OPT{}$, where $\OPT{}$ is the minimum average degree of any navigable graph on $P$ with respect to the Euclidean metric.
\end{theorem}

\subsection{Hardness for Low-dimensional Euclidean Metrics}\label{subsec:LowEuclideanHardness}
Our next goal is to establish hardness of approximation results for constructing navigable graphs in low-dimensional Euclidean space, specifically when $d$ is polylogarithmic in $n$. The main idea is to construct a point set similar to the one we saw in the previous section project it down to a low dimensional space using a JL \cite{DasguptaGupta:2003} transform, while maintaining it being a covering point set. 
Note that the point set presented in \cref{subsec:EuclideanHardness} is not usable for this purpose - for example if we want to maintain property (3) we have to make sure whatever distortion factor $\varepsilon$ we use will maintain: $(1-\varepsilon)\cdot (n^2+f_i+1) > (1+\varepsilon)\cdot (n^2+f_i)$ which would generally require $\varepsilon$ to be inverse-polynomial, thus the projection will land in polynomial-dimensional space.
Our way of solving this issue is looking at set cover instances with bounded-frequency, and using known hardness of approximation results for those instances.
We start by defining, for set cover instances with bounded frequencies, covering point sets which maintain their properties under JL transform with a distortion factor of $\Theta\left(\frac{1}{k^2}\right)$ where $k$ is an upper bound on the frequency.
Given a set cover instance $(U,\mathcal{F})$ we define a parameterized mapping $\varphi(U,\mathcal{F},f)$ which maps the set cover instance to a representative point set $P\subseteq \mathbb{E}^{(m+n+2+(m+n)^{3})}$ with $(m+n)^3$ solution points, which we denote by $P_{SC}=\{\ell_1,...,\ell_{(m+n)^3}\}$, as follows:
\begin{itemize}
    \item $\forall S_{j}\in\mathcal{F}:p_{S_{j}}=\mathbf{e}_{j}+\mathbf{e}_{n+m+1}$.
    \item $\forall i\in\left[n\right]:p_{i}=f\cdot \mathbf{e}_{m+i}+\sum_{j:i\in S_{j}}\mathbf{e}_{j}$
    \item $\forall k\in\left[\left(m+n\right)^{3}\right]:\ell_{k}=\frac{1}{2}\cdot \mathbf{e}_{n+m+2}+\sqrt{0.65}\cdot \mathbf{e}_{n+m+2+k}$
    \item $c=\mathbf{e}_{n+m+2}$.
\end{itemize}
\cref{fig:low_dim_point_distances} shows a summary of the distances between the different types of points in the representative point set we've constructed in table form.
\begin{center}
\setlength{\fboxrule}{0.8pt} % border width (unchanged by scaling)
\setlength{\fboxsep}{8pt}    % gap between border and table

\begin{figure}[htbp]
  \centering
    \scalebox{0.9}{%  ← tweak this number
      \begin{tabular}{ c|c|c|c|c } 
        
         & Element points & Set points & Solution points & $c$ \\
        \hline
        Element points &
          $\geq\sqrt{2f^{2}}$ &
          $d\!\left(p_{i},p_{S_{j}}\right)=
            \begin{cases}
              \sqrt{f^{2}+f_i} & i\in S_{j}\\
              \sqrt{f^{2}+f_i+2} & \text{otherwise}
            \end{cases}$ &
          $\sqrt{f^{2}+f_i+0.9}$ &
          $\sqrt{f^{2}+f_i+1}$ \\ \hline
        Set points       & "               & $\sqrt{2}$   & $\sqrt{2.9}$ & $\sqrt{3}$   \\ \hline
        Solution points  & "               &         "     & $\sqrt{1.3}$ & $\sqrt{0.9}$ \\ \hline
        $c$              &    "            &      "        &   "           & 0            \\ 
      \end{tabular}
    }% end \scalebox
  \caption{Distances between different types of points in $\varphi(U,\mathcal{F},f).$}
  \label{fig:low_dim_point_distances}
\end{figure}
\end{center}
\begin{lemma}\label{lem:distortion_resistant_covering_point_set}
If $(U,\mathcal{F})$ is a set cover instance in which the frequency of every element is exactly $k\geq 2$, and let $d$ be a distance metric on $P=\varphi(U,\mathcal{F},k)$ that satisfies: 
\[
\forall p\neq p'\in P:d(p,p')\in 
\left[(1-\delta)\cdot \lVert p-p' \rVert,(1+\delta)\cdot \lVert p-p' \rVert\right]
\]
where $\lVert \cdot \rVert$ denotes the euclidean norm, then: If $\delta \leq \frac{1}{40k^2}$, $P$ is a covering point set for $(U,\mathcal{F})$ w.r.t $d$.
\end{lemma}
\begin{proof}
First note, that if the frequency of every element in the set cover instance is $k$, the only changes in distances in $\varphi(U,\mathcal{F},k)$ are in the category of element points.
Specifically, the distance between $c$ and an element point is now $\sqrt{k^2+k+1}$, the distance between a solution point and an element point is $\sqrt{k^2+k+0.9}$.
The distance between an element point and a set point is $\sqrt{k^2+k}$ or $\sqrt{k^2+k+2}$, depending on whether the element is a member of the set.
Now, to verify $\varphi(U,\mathcal{F},f)$ is a covering point set w.r.t $d$:
Note that:
\[
0.9\cdot(1+\delta)\leq 0.9+\frac{0.9}{40}< \sqrt{1.3}-\frac{\sqrt{1.3}}{40}\leq \sqrt{1.3}\cdot(1-\delta)
\]
This implies that $c$ is still strictly closer to a solution point than a different a solution point is. The Euclidean distance between solution points to other classes of points is only greater, which implies $c$ is still the unique nearest neighbor of all solution points, so property (1) is maintained.
Now, observe:
\[
\frac{\left(\sqrt{k^{2}+k+0.9}+\sqrt{k^{2}+k+1}\right)^{2}}{0.1}\leq\frac{4k^{2}}{0.1}=40k^{2}.
\]
Thus:
\[
\frac{\sqrt{k^{2}+k+0.9}+\sqrt{k^{2}+k+1}}{\sqrt{k^{2}+k+1}-\sqrt{k^{2}+k+0.9}}=\frac{\left(\sqrt{k^{2}+k+0.9}+\sqrt{k^{2}+k+1}\right)^{2}}{\left(\sqrt{k^{2}+k+1}-\sqrt{k^{2}+k+0.9}\right)\left(\sqrt{k^{2}+k+0.9}+\sqrt{k^{2}+k+1}\right)}
\]
\[
\frac{\left(\sqrt{k^{2}+k+0.9}+\sqrt{k^{2}+k+1}\right)^{2}}{0.1}\leq 40k^2.
\]
So:
\[
\frac{\sqrt{k^{2}+k+0.9}+\sqrt{k^{2}+k+1}}{40k^2}\leq\sqrt{k^{2}+k+1}-\sqrt{k^{2}+k+0.9},
\]
which implies
\[
\sqrt{k^{2}+k+0.9}\cdot \left(1+\frac{1}{40k^2}\right)\leq\sqrt{k^{2}+k+1}\cdot \left(1-\frac{1}{40k^2}\right)
\]
\[
\to \sqrt{k^{2}+k+0.9}\cdot \left(1+\delta\right)\leq\sqrt{k^{2}+k+1}\cdot \left(1-\frac{1}{\delta }\right).\qedhere
\]
\end{proof}
This implies that given an element point $p_i$ and a solution point $\ell_r$, $c$ is still further from $p_i$ than $\ell_r$ is. Further, any other point $p\in \varphi(U,\mathcal{F},k)$ that does not correspond to a set containing $i$ is further from $p_i$ than $c$ is in the Euclidean metric, so they also remain further from $p_i$ than $\ell_r$ is under $d$. For a point corresponding to a set containing $i$, i.e of the form $p_{S_j}$ where $i\in S_j$, it can similarly be verified that:
\[
\sqrt{k^{2}+k}\cdot \left(1+\delta\right)\leq\sqrt{k^{2}+k+0.9}\cdot \left(1-\frac{1}{\delta }\right).
\]
This means moving from a solution point to $p_{S_j}$ where $i\in S_j$ gets us closer to $p_i$ under $d$. Overall property (3) is maintained.
The other properties can be similarly verified.
\subsubsection{Hardness in Polylogarithmic Dimensions}
\begin{theorem} There exist constants $a,b>0$ such that for every $\varepsilon >0$, unless $NP\subseteq DTIME\left(n^{O(\log\log(n))}\right)$, there is no polynomial-time algorithm that, given a set $P$ of $n$ points in $d$-dimensional Euclidean space where $d\in O\left(\left(\log n\right)^{1+\varepsilon }\right)$, constructs a navigable graph with average degree is $\leq \left(\frac{1}{2a^{\kappa_\varepsilon}}\cdot  \left(\log(n)\right)^{\kappa_\varepsilon}-0.1\right) \cdot \OPT{}$, where $\kappa_\varepsilon=\min\left\{\frac{1}{b},\frac{\varepsilon}{4}\right\}$ and $\OPT{}$ is the minimum average degree of any navigable graph on $P$ with respect to the Euclidean metric.
\end{theorem}
\begin{proof}
Let $\varepsilon>0$ and assume that for every $a,b>0$, such an algorithm exists. We show that this implies a polynomial-time algorithm that produces a $\frac{1}{2}\cdot\log(N)-0.01$ approximation for the $E_k$-vertex-cover instances produced in the reduction in Section 6.2 of \cite{DinurGuruswami:2005}, where $N$ is the number of hyperedges, which implies $NP\subseteq DTIME\left(n^{O(\log\log(n))}\right)$. Note that within the reduction, setting $k=\left(\log(n)\right)^\frac{1}{b'}$ for a sufficiently large constant $b'$ yields hard instances in which the number of hyperedges $N$ is equal to $n^{O(\log\log(n))}$, the number of vertices $M$ is at most $n^{O(\log\log(n))}$, and the size of a minimum vertex cover is at least $n^{O(\log\log(n))}$ (since the degree is bounded by $n$), and in particular with this value of $k$ we ensure $M+N+1+(M+N)^3\leq N^{a'}$ for some constant $a'$.
We set $b=b'$ and $a=a'$.
Now, set $k=\log^{\kappa_\varepsilon}(n)$ in the reduction and let $G$ be a hypergraph constructed by the reduction. The set cover instance $\left(U_{G}=E,\mathcal{F}_{G}=V\right)$ induced by $G$, in which for every $e\in E$ and $v\in V$ we say that $e\in v$ if and only if $v$ is adjacent to $v$, has $N$ elements and $M$ sets. Further, the frequency of every element is exactly $k$. Note that $V^{\prime}\subseteq V$ is a vertex cover for $G$ if and only if it is a set cover for $\left(U_{G},\mathcal{F}_{G}\right)$. 
We construct $\varphi(U_G,\mathcal{F}_G,k)$, and apply a deterministic polynomial-time JL transform \cite{DasguptaGupta:2003} to the result with distortion parameter $\delta=\frac{1}{40k^2}$, obtaining a projected point set $P'\in \mathbb{E}^{\log(n)\cdot\Theta(k^4)}=\mathbb{E}^{\Theta\left(\left(\log n\right)^{1+\varepsilon}\right)}$.
Since the frequency of every element in $U_G$ is exactly $k$, according to \Cref{lem:distortion_resistant_covering_point_set}, $P'$ is a covering point set for $(U_G,\mathcal{F}_G)$. We now apply the hypothesized algorithm to $P'$ and obtain a navigable graph $H_{P'}$ whose degree is $(\frac{1}{2a^{\kappa_\varepsilon}}\cdot  \left(\log(N')\right)^{\kappa_\varepsilon}-0.1)$-approximate minimum, where $N'$ is the number of points in $P'$. Observe that $N'=N+M+1+(M+N)^3\leq N^{a}$.
Thus, the degree is $r$-approximate minimum where
$r\leq (\frac{1}{2a^{\kappa_\varepsilon}}\cdot  \left(\log(N^a)\right)^{\kappa_\varepsilon}-0.1)=\frac{1}{2}\cdot \log(N)-0.1$.
Now, according to \Cref{lem:solutionForCoveringPointSetsSolvesSetCover}, we can use $H_{P'}$ to efficiently construct a set cover for $(U_G,\mathcal{F}_G)$ that is $\left(r+o(1)+\Theta(\frac{r}{OPT})\right)$-approximate minimum.
Note:
\[
r+o(1)+\Theta\left(\frac{r}{OPT}\right)
\leq r+o(1)+\Theta\left(\frac{r}{n^{O(\log\log(n))}}\right)\leq r+o(1)
\]
\[
=\frac{1}{2}\cdot \log(N)-0.1+o(1)\leq \frac{1}{2}\cdot \log(N)-0.01
\]
Thus, we have an efficient algorithm that produces a $\left(\frac{1}{2}\cdot \log(N)-0.01\right)$-approximate minimum solution.
\end{proof}
\subsubsection{Hardness in $\Theta(\log(n))$ Dimensions}
\begin{theorem}
For every constants $k\geq 2,\varepsilon >0$, there exists $C(k)$, a function of $k$, such that unless $P=NP$, there does not exist a polynomial-time algorithm that, given a set $P$ of $N$ points in 
%$E^{C(k)\cdot\log(N)}$, i.e 
$(C(k)\cdot \log(N))$-dimensional Euclidean space, returns a navigable graph on $P$ with average degree $\leq (k-1.5-\varepsilon) \cdot \OPT{}$, where $\OPT{}$ is the minimum average degree of any navigable graph on $P$ with respect to the Euclidean metric.
\end{theorem}
\begin{proof}
We assume such an algorithm exists, and then use it to construct a polynomial time algorithm that produces a $(k-1-\varepsilon)$-approximate solution for the $E_k$-vertex-cover problem, which according to section 5 of \cite{DinurGuruswami:2005} would imply $P=NP$.
Let $k\geq2$, $\varepsilon>0$, and set $\varepsilon'=1.5\varepsilon>\varepsilon$.
Let $G=\left(V,E\right)$ be an $E_{k'}$-vertex-cover instance that is hard to $(k'-1-\varepsilon')$-approximate for $k'=2k$ as constructed in the reductions proving hardness of approximation for hypergraph vertex cover by \cite{DinurGuruswami:2005}. Let $n$ denote the number of vertices in $G$, and $m$ denote the number of edges.
Now, consider the set cover instance induced by $G$: $\left(U_{G}=E,\mathcal{F}_{G}=V\right)$ in which for every $e\in E$ and $v\in V$ we say that $e\in v$ if and only if $v$ is adjacent to $v$, it has $m$ elements and $n$ sets. Further, the frequency of every element is exactly $k'$. Note that $V^{\prime}\subseteq V$ is a vertex cover for $G$ if and only if it is a set cover for $\left(U_{G},\mathcal{F}_{G}\right)$. Now, we construct $P=\varphi\left(U_{G},\mathcal{F}_{G},k'\right)$ and apply a deterministic polynomial-time JL transform \cite{DadushDaniel:SparseJL} to the result with distortion parameter $\delta=\frac{1}{40k^{2}}$, and obtain a projected point set $P'\in\mathbb{E}^{\log(N)\cdot \Theta(k^4)}$, i.e in $\log(N)\cdot C(k)$ dimensions where $C(k)$ is a constant depending on $k$, and $N$ is the number of points in $P'$. Note that $N=m+n+\left(m+n\right)^{3}+1$, so $N$ is  polynomial in $\left(m,n\right)$. According to \Cref{lem:distortion_resistant_covering_point_set}, since the frequency of every element in $U_{G}$ is exactly $k'$, $P'$ is a covering point set for $\left(U_{G},\mathcal{F}_{G}\right)$ with respect to the Euclidean distance metric.
We can now apply the hypothesized algorithm to $P'$ to obtain a navigable graph (in time polynomial in $N$, and thus in $m,n$) $H_{P'}$ whose degree is $r$-approximate minimum where $r=k-1.5-\varepsilon$. Note that according to \Cref{lem:solutionForCoveringPointSetsSolvesSetCover} we can then use $H_{P'}$ to construct a set cover for $\left(U_{G},\mathcal{F}_{G}\right)$ whose average degree is $\left(r+o\left(1\right)+\Theta\left(\frac{r}{\OPT{}}\right)\right)$-approximate minimum. Observing the proof of \Cref{lem:solutionPointsMakeSolutions} and specifically \Cref{eq:reductionApproxRatio}, we can actually conclude that $H_{P'}$ can be used to efficiently compute a $\left(2r+2+o\left(1\right)\right)$-approximate minimum set cover, because:
\[
\left(r+o\left(1\right)+\frac{1}{\OPT{}}\right)\cdot\left(\OPT{}+1\right)=\left(r+o\left(1\right)+\frac{1}{\OPT{}}+\frac{r}{\OPT{}}+\frac{1}{\OPT{}^{2}}\right)\cdot \OPT{}\leq\left(2r+o\left(1\right)+2\right)\cdot \OPT{}
\]
Overall, we have a $\left(2r+o\left(1\right)+2\right)$-approximate minimum set cover for $(U_G,\mathcal F_G)$. 
Note that:
\[
2r+o\left(1\right)+2=2k-3-2\varepsilon+o\left(1\right)+2=k'-1-2\varepsilon+o\left(1\right)\leq k'-1-\varepsilon'
\]
Thus, we have a $(k'-1-\varepsilon')$-approximate minimum solution in polynomial time for the hard $E_{k'}$-vertex cover instance we started with, which according to \cite{DinurGuruswami:2005} implies $P=NP$.
\end{proof}

\subsubsection{Concluding Remarks}
We conjecture that the problem of constructing minimum navigable graphs is $\Theta(\log(n))$-hard to approximate even when restricting to polylogarithmic-dimensional Euclidean spaces. We note that this is related to a conjecture previously presented in the literature - Conjecture 6.1 of \cite{DinurVenkatesanKhot}. If that conjecture is proven, the same techniques we used in our reductions can likely be applied to its hard instances to
show a hardness of approximation factor of $\Theta(\log(n))$ in polylogarithmic-dimensional Euclidean space.

\section{Applications to Other Graph Construction Problems}
\label{sec:generic_applications}
Beyond constructing navigable graphs, in this section we argue that sublinear-time set cover results, including \Cref{thm:general_sublinear} or results from prior work \cite{IndykMahabadiRubinfeld:2018}, can also be used to obtain fast constructions for $\alpha$-shortcut reachable and $\tau$-monotonic graphs, which have similarly received recent attention due to their applications in nearest neighbor search \cite{SubramanyaDevvritKadekodi:2019,IndykXu:2023,PengChoiChan:2023}.
As discussed in \Cref{sec:nav_graph_set_cover}, both of these graph construction problems can be formulated as $n$ simultaneous set-cover instances, one for each point in a dataset $P = \{p_1, \ldots, p_n\}$.

As before, for the instance corresponding to node $p_i$, the elements that need to be covered are all other points $p_k$, $k \neq i$. We also have a set corresponding to each $p_j$, $j \neq i$.
For $\alpha$-shortcut reachability we will call this set $A_{i\rightarrow j}$, and for $\tau$-monotonicity, we'll call it $T_{i\rightarrow j}$.
$A_{i\rightarrow j}$ contains any point $p_k$ such that $d(p_j,p_k) < \frac{1}{\alpha} d(p_i,p_k)$, and $T_{i\rightarrow j}$ contains any point $p_k$ such that $d(p_j,p_k) < d(p_i, p_k) - \tau$, as well as $p_j$ itself.\footnote{If $d(p_i, p_k) < \tau$, it is not possible to find a $j$ with $d(p_j,p_k) < d(p_i, p_k) - \tau$. The definition of $\tau$-monotonicity requires that, in this case, $p_i$ is directly connected to $p_k$. Equivalently, we can add $p_k$ to $T_{i\rightarrow k}$ to ensure that an edge to $p_k$ covers $p_k$.}
A graph is $\alpha$-shortcut reachable (resp. $\tau$-monotonic) if and only if for each $p_i$, the collection $A_{i\rightarrow j}$ (resp. $T_{i \to j}$), taken over outneighbors $j$, forms a set cover for $p_i$'s instance.

To leverage \Cref{thm:general_sublinear}, we need to provide efficient implementations of the $\MemberOf$, $\FreqOf$, and $\SetOf$ queries for the corresponding set cover problems. $\MemberOf$ queries are trivial -- checking if a point is contained in $A_{i\rightarrow j}$ or $T_{i\rightarrow j}$ just requires a single distance comparison. 
However, we need to slightly modify our approach to FreqOf and SetOf queries. As in \Cref{clm:sublinear_access}, we build a distance-based permutation matrix in $O(n^2 \log n)$ time. However, we no longer answer a $\FreqOf(p_i,p_j)$ query by returning $\pi_j^{-1}(i) - 1$. Instead, we precompute the answer to FreqOf for every query pair $p_i, p_j$ and store them in a lookup table. To efficiently compute the answers to FreqOf queries, we make the following observation: as for the standard navigability problem, all $A_{i \to j}$ (or equivalently $T_{i \to j}$) containing $p_k$ are given by a prefix of $\pi_k$. We can find the length of this prefix via binary search on $\pi_k$ for the rightmost $j$ satisfying $d(p_j, p_k) < \frac{1}{\alpha}d(p_i, p_k)$ (or analogously $d(p_j, p_k) < d(p_i, p_k) - \tau$ for $\tau$-monotonicity). A binary search suffices here because $\frac{1}{\alpha}d(p_i, p_k)$ is fixed for a given pair $p_i, p_k$, and $\pi_k$ stores all $j$ in sorted order by $d(p_j, p_k)$. The prefix length obtained from this binary search gives us the answer to the corresponding FreqOf query. Thus, we can precompute the answers to all $n^2$ possible FreqOf queries in $O(n^2 \log n)$ time. These FreqOf queries then trivially give us SetOf queries.

With $\MemberOf$, $\FreqOf$, and $\SetOf$ now available in $O(1)$ time, we can immediately apply \Cref{thm:general_sublinear} to obtain our main corollary on $\alpha$-shortcut reachable and $\tau$-monotonic graphs.

\begin{proof}[Proof of \Cref{cor:alpha-tau}]
Consider first the $\alpha$-reachability problem. Let $\OPT_1^\alpha, \ldots, \OPT_n^\alpha$ denote the optimum value for each minimum set cover instance, so that the minimum total number of edges in $G$ is $\OPT^\alpha = \sum_{i=1}^n \OPT_i^\alpha$. Applying \Cref{thm:general_sublinear}, we obtain an $O(\log n)$-approximation to every set cover instance in time:
\begin{align*}
\tilde{O}\left(\sum_{i=1}^n \min\left(n\cdot \OPT_i^\alpha, \frac{n^2}{\OPT_i^\alpha}\right)\right) \leq \tilde{O}\left(\sum_{i=1}^n \min\left(n\cdot \OPT_i^\alpha, n^{1.5}\right)\right) = \tilde{O}\left(\min\left(n \cdot \OPT^\alpha, n^{2.5}\right)\right)
\end{align*} 
The inequality follows from noting that one of $n^2/\OPT_i^\alpha$ or $n\cdot OPT_i^\alpha$ is always less than $n^{1.5}$.

The proof for $\tau$-monotonicity is identical. We note that for both problems we obtain the strong guarantee of nearly minimizing the degree of \emph{every node simultaneously}, so we return a graph with both average degree and maximum degree within $O(\log n)$ of optimal. 
\end{proof}

As a final note, we point out that the techniques that we applied in Algorithm~\ref{alg:mng-full} to improve on $\tilde O(n^{2.5})$ to $\tilde O(n^2)$ for the navigable graph problem do not directly apply to either $\alpha$-shortcut reachability or $\tau$-monotonicity  (see \Cref{sec:prelims} for a summary of these techniques). In particular, adding random edges or fixed cliques is not guaranteed to significantly reduce the number of uncovered elements across all $n$ set cover instances for these stronger versions of navigability. It is interesting to ask if it is possible to further improve on \Cref{cor:alpha-tau}, possibly with other techniques.

%% file: genericAlgorithms.tex
\section{Sublinear Set Cover Algorithm}

\newcommand{\AddSet}{\textsc{AddSet}\xspace}
\newcommand{\GenVoteCover}{\textsc{GeneralVoteCover}\xspace}
\newcommand{\EarlyTermGVC}{\textsc{EarlyTerminationGVC}\xspace}
\newcommand{\LimVoteCover}{\textsc{LimitedVoteCover}\xspace}
\newcommand{\LazyETGVC}{\textsc{LazyEarlyTerminationGVC}\xspace}
\newcommand{\LazyLimVoteCover}{\textsc{LazyLimitedVoteCover}\xspace}

\label{sec:sublinear-set-cover}
In this section, we prove that the techniques presented in \Cref{sec:main_algo} can be applied to the more general problem of solving set cover in sublinear time. In particular, we show:
\begin{theorem}
    \label{thm:general_sublinear}
Suppose that we have an unweighted minimum set cover instance with $n$ elements, $m$ sets, and minimum value $\OPT$, and that we can issue SetOf and Membership queries in $O(1)$ time. Then, there is a randomized simulation of the greedy set cover algorithm that, with high probability, produces an $O(\log n)$ approximation to the set cover instance in 
$\tilde{O}\left(\min\left(m + n\cdot \OPT, \frac{mn}{\OPT}\right)\right)$
 time.
\end{theorem}
We stress that, while our approach is new, the bound of \Cref{thm:general_sublinear} is not: as discussed in \cite{IndykMahabadiRubinfeld:2018}, an unpublished modification of \cite{KoufogiannakisYoung:2014} combined with another algorithm from \cite{IndykMahabadiRubinfeld:2018} can achieve the same runtime \cite{Vakilian:2025}. Among more easily accessible prior results, 
\Cref{thm:general_sublinear} can be compared to Theorem 4.4 in \cite{IndykMahabadiRubinfeld:2018}, which for any $\alpha > 1$, obtains an $O((\alpha + 1) \log n)$ approximation in almost the same time, but with $m + n\cdot \OPT$ replaced with $m\left(\frac{n}{\OPT}\right)^{1/\alpha} + n\cdot \OPT$.\footnote{Technically, \cite{IndykMahabadiRubinfeld:2018} combines SetOf queries with ``Element Of'' (\EltOf) queries instead of Membership queries. While these queries are incomparable (neither is stronger than the other), the algorithms from  \cite{KoufogiannakisYoung:2014} and \cite{IndykMahabadiRubinfeld:2018} can be easily modified to work with Membership queries, and our algorithm can be modified to work with EltOf queries without changing the runtime.} The fact that these previous results can be matched with a greedy-type algorithm may be of interest.  Indeed, our approach is similar to graph construction algorithms used in practice: the ``slow preprocessing'' variant of the Vamana/DiskANN also constructs a navigable graph by greedily building a set cover for each node, although a heuristic rule is used to select which set to added to the current solution at each step \cite{SubramanyaDevvritKadekodi:2019}. 

% We will then establish \Cref{cor:alpha-tau} by demonstrating how the problems of constructing $\alpha$-shortcut reachable graphs, and constructing $\tau$-monotonic graphs can be reduced to the sublinear set cover problem.
Before proving \Cref{thm:general_sublinear}, we first review the general task. The problem of finding a minimum set cover is a classic \npcomplete problem that has been studied in a variety of settings.
The input to a set cover problem includes a \defn{set system}, which is a pair $(U=\{e_1,...,e_n\},\mathcal{F}=\{S_1,...,S_m\})$ where each $S_i\subseteq U$.
The output to the minimum set cover problem is a minimum subset of $\mathcal{F}$ which includes a set covering each $e_i\in U$. 
While set cover is \npcomplete, the greedy algorithm which chooses the set covering the maximum number of remaining elements, is a $O(\log n)$ approximation which runs in $O\left(\sum_{i=1}^m\left|S_i\right|\right)$ time.
Unless $\p=\np$, this is essentially the best possible approximation ratio for a polynomial-time algorithm~\cite{DinurSteurer:2014}.

Our results are in a sublinear model of the problem, where we do not explicitly interact with $U$ and $\mathcal{F}$, but rather through specific queries which we assume can be performed in constant time.
Specifically, we assume access to the \SetOf query described in Section~\ref{sec:prelims}, which returns for a given element $e \in U$ all sets in $\mathcal{F}$ containing $e$. In addition, we assume access to the related \EltOf query, which returns for a given set $S \in \mathcal{F}$ all elements contained in $S$. Note that with $O(\log n)$ calls to $\EltOf$, we can also evaluate a $\FreqOf$ queries for each element. 

% Similar to what is demonstrated in the proof of \Cref{clm:sublinear_access}, we can also support the $\FreqOf$ query with $O(n\log n)$ additional preprocessing time.

We now present an emulation of the greedy algorithm in the sublinear model, then progressively modify it to achieve two algorithms which will together provide a proof for \cref{thm:general_sublinear}.

\subsection{The greedy algorithm in the sublinear model}
We first present an implementation of the classical greedy algorithm in the sublinear model, \Cref{alg:GreedySetCover}, which we will later use as a subroutine.

\begin{algorithm}
\caption{GreedySetCover($U,\mathcal{F}$)}
\label{alg:GreedySetCover}
\begin{algorithmic}[1]
\State $C\gets \varnothing$.
\State Initialize an empty max priority queue $P$.
\ForEach{$S\in \mathcal{F}$}
\State Compute $|S|$ using \EltOf queries.
\State Insert $S$ into $P$ using $|S|$ as the priority key.
\EndFor
\While{$U\neq \varnothing $}
\State Pop the top priority set $S$ from $P$.
\State $U\gets U\setminus S, C\gets C\cup \{S\}$.
\State Enumerate the elements in $S$ using \EltOf queries.
\ForEach{$e\in S$}
\State Enumerate all sets $S^\prime$ such that $e\in S^\prime$ using $\SetOf$ queries.
\State For each such set $S^\prime$, decrease the key of $S^\prime$ in $P$ by $1$. 
\EndFor
\EndWhile
\State \Return{C}
\end{algorithmic}
\end{algorithm}

\begin{lemma}
    \Cref{alg:GreedySetCover} runs in $O\left(\sum_{S\in\mathcal{F}}|S|\cdot \log(m)\right)$ time, and achieves a $\ln(n)$ approximation ratio.
\end{lemma}

\begin{proof}
    The algorithm emulates the greedy algorithm in the standard model and therefore obtains the same approximation ratio.
    We use a standard binary heap with $O(\log n)$-time operations as the priority queue.
    The running time is determined by the time spent performing heap operations, as all other costs are linear in $n,m$.
    The algorithm performs at most $m$ pop operations on the heap, which requires $O(m\cdot \log(m))$ time. 
    For each set $S$ that is popped from the heap the algorithm perform $|S|$ decrease key operations, and each set is popped at most once.
    Thus, the total number of decrease key operations is bounded by $\sum_{S\in\mathcal{F}}|S|$, and thus requires at most $O\left(\sum_{S\in\mathcal{F}}|S|\cdot \log(m)\right)$ time. 
    The total running time is therefore bounded by $O\left(\sum_{S\in\mathcal{F}}|S|\cdot \log(m)\right)$.
\end{proof}

\subsection{The voting algorithm}
\label{subsec:vote-cover-analysis}
As in \Cref{alg:mng-cover}, the first approach to improving the run time of the greedy algorithm is to approximate the greedy choice made at each step by using voting.
To do this, the algorithm samples a random element from the pool of unsampled uncovered elements (prospective voters), $N$, which then votes for all sets it is an element of.
The element is added to the set of active voters, $V$.
Whenever a set reaches a logarithmic number of votes, the algorithm chooses it as the next set, which involves enumerating all of its elements using membership queries, removing them from $P$ and $V$, as well as removing any votes they may have cast.
This process continues until $N$ is empty, at which point the set of uncovered elements is $V$, so the residual instance is relatively small, and we can just use the greedy algorithm.

\begin{algorithm}
\caption{GeneralVoteCover($U,\mathcal{F}$)}
\label{alg:ConstructVoteCover}
\begin{algorithmic}[1]
\State $N\gets U$, $V \gets \varnothing$, $C\gets \varnothing$
\While {$N\neq\varnothing$}
\State Sample a random element from $N$, denoted by $j$.
\State $N\gets N\setminus{j}$
\State $V\gets V\cup{j}$
\State Use membership queries to enumerate all sets in $\mathcal{F}$ that $j$ is a member of, denote this set by $\mathcal{F}[j]$.
\ForEach{$S\in F[j]$}
\State Increment the vote count for $S$.
\If{the number of votes for $S$ is $100\cdot(\log(m)+\log(n))$}.
\State $\AddSet(S)$
%\If{$|C|=L$}
%\State \Return{\FAIL}
%\EndIf
\EndIf
\EndFor
\EndWhile
\State $C^\prime\gets  GreedySetCover(V,\{S_{\overline{C}}=S\cap U_{\overline{C}}:S\in \mathcal{F}\setminus C\})$
%\If{$|C^\prime|\geq L$}
%\State \Return{\FAIL}
%\EndIf
\State $C\gets C\cup C^\prime$
\State \Return{$C$}
\vspace{0.5em}
\hrule
\vspace{0.5em}
\Function{\AddSet}{$S$}
    \State Use membership queries to enumerate all elements in $S$.
    \ForEach{$j\in S$}
    \If{$j \in V$}
    \State Use membership queries to enumerate all sets in $\mathcal{F}$ that $j$ is a member of.
    \State For each such set, decrement the vote count.
    \State Remove $j$ from $V$. 
    \Else
    \State Remove $j$ from $N$.
    \EndIf
    \EndFor
    \State $C\gets C\cup\{S\}$
\EndFunction
\end{algorithmic}
\end{algorithm}

\paragraph{Notation.}
For ease of presentation we introduce the following notation: $U_{\overline{C}}=U\setminus \left(\bigcup_{S\in C}S\right)$, i.e we let $U_{\overline{C}}$ denote the set of all elements in $U$ that are uncovered by the sets in $C$.
We let $\mathcal{F}_{\overline{C}}$ denote the following collection of sets: $\{S_{\overline{C}}=S\cap U_{\overline{C}}:S\in \mathcal{F}\setminus C\}$, i.e the projection of all sets not yet chosen onto $U_{\overline{C}}$.
We let $\mathcal{S}_{\overline{C}}$ denote the following set system: $(U_{\overline{C}},\mathcal{F}_{\overline{C}})$. 
This set system describes the residual set cover instance we've left to solve after choosing the sets in $C$ - all uncovered elements, and the projection of all unchosen sets to the set of uncovered elements.
We denote the number of votes a set $S$ received (i.e the value of its vote counter) by $v(S)$. 
\begin{claim}\label{obs:OPTMaintenance}
If $\mathcal{S}=(U,\mathcal{F})$ has a minimum set cover of size $\OPT$, then for any $C\subseteq \mathcal{F}$, then $\mathcal{S}_{\overline{C}}$ has a set cover which size is at most $\OPT$.
\end{claim}

\subsubsection{Bounding the approximation ratio.}
The key difference between the voting algorithm and the greedy algorithm is that the voting algorithm approximates the largest sets at each step.
Therefore, the key thing to prove in order to bound the approximation ratio is that these approximations are sufficient. We begin by observing that \Cref{alg:ConstructVoteCover} maintains the following invariants:

\begin{claim}
\label{genAlgorithmInvariants}
The main while loop in \GenVoteCover maintains the following invariants:
\begin{enumerate}
    \item $N\cup V=U_{\overline{C}}$,
    \item $V$ is a set of random, uniformly distributed elements in $U_{\overline{C}}$,
    \item For every set $S$: $v(S)=\lvert V\cap S_{\overline{C}}\rvert$, and
    \item For every set $S$: $v(S)\leq 100(\log(m)+\log(n))$.
\end{enumerate}
\end{claim}

Next we show that the set chosen at each step is the largest set up to constant factors with high probability.

\begin{lemma}\label{lem:approxGreedyPicksApproxLargestSet}
With probability at least $1-\frac{1}{m^{28}\cdot n^{28}}$, in every call to $\AddSet(S)$, $S_{\overline{C}}$ is approximately the largest set in $\mathcal{F}_{\overline{C}}$, more specifically:
\[
\forall S^\prime \in F\setminus C: \lvert S_{\overline{C}}\vert \geq\frac{1}{8}\cdot \lvert S^\prime_{\overline{C}} \rvert 
\]
\end{lemma}
\begin{proof}
It suffices to bound the probability that at some point in the execution there exists a pair $S,S^\prime \in \mathcal{F}\setminus C$ such that $v(S)= 100(\log(m)+\log(n))$ and $\lvert S_{\overline{C}}\vert <\frac{1}{8}\cdot \lvert S^\prime_{\overline{C}} \rvert$.
Assume that during the execution of the algorithm, such a pair exists.
By \Cref{genAlgorithmInvariants}, $v(S^\prime)\leq 100(\log(m)+\log(n))$ which in turn implies, $\lvert V\cap S_{\overline{C}}\rvert\geq 100\cdot(\log(m)+\log(n))$ while $\lvert V\cap S_{\overline{C}}^\prime\rvert\leq100 (\log(m)+\log(n))$.
Because $V$ is uniformly random over $U_{\overline{C}}$, we can apply a standard Chernoff bound to conclude the probability of $\lvert V\cap S_{\overline{C}} \rvert$ exceeding $100\cdot (\log(m)+\log(n))$ is at most $\frac{1}{m^{30}\cdot n^{30}}$.
\end{proof}
\noindent We now can use this to bound the number of sets added in the main loop.
\begin{lemma}\label{lem:numberOfCallsToAddSet}
With probability at least $1-\frac{1}{m^{28}\cdot{n^{28}}}$, the algorithm calls $\AddSet(S)$ at most $O(\OPT\cdot \log(n))$ times.
\end{lemma}
\begin{proof}
    Consider some optimal solution $\mathcal{F}^\prime \subseteq \mathcal{F}$ to the input set cover instance, and denote its size by $\OPT$.
    Note that for every subset $U^\prime \subseteq U$, since $\mathcal{F}^\prime$ covers $U^\prime$ there exists a set $S^\prime \in \mathcal{F^\prime}$ that covers a $\frac{1}{\OPT}$-fraction of all elements in $U^\prime$.
    Consider a call to $\AddSet(S)$: We know that there exists some set in $S^\prime\in \mathcal{F^\prime}$ that covers a $\frac{1}{\OPT}$ fraction of all elements in $U_{\overline{C}}$, meaning $\lvert S^\prime_{\overline{C}}\rvert \geq \frac{1}{\OPT}\cdot \lvert U_{\overline{C}}\rvert$.
    Note that $S^\prime\in \mathcal{F}\setminus C$ since for every set in $A\in C$ we have $A\cap U_{\overline{C}}=\varnothing$.
    By \Cref{lem:approxGreedyPicksApproxLargestSet}, with probability of at least $\frac{1}{m^{28}\cdot n^{28}}$ we have $\lvert S_C\rvert\geq \frac{1}{8}\lvert S^\prime_C\rvert$ and thus $\lvert S_C\rvert \geq \frac{1}{8\cdot \OPT}\cdot \lvert U_C\rvert$.
    So each call to $\AddSet$ covers a $\frac{1}{8\cdot \OPT}$ fraction of the remaining uncovered elements with high probability.
    Therefore, at most $8\ln{n}\cdot \OPT$ calls can be made.
\end{proof}
\noindent Finally, we observe that the call to the greedy algorithm cannot add more then $\ln(n) \cdot \OPT$ sets, so we have:
\begin{claim}\label{lem:cvcApproximationRatio}
    With probability at least $1-\frac{1}{m^{28}\cdot{n^{28}}}$, \Cref{alg:ConstructVoteCover} returns a $9\ln(n)$-approximate minimum set cover.
\end{claim}

\paragraph{Running time analysis.}
To analyze the running time, we first observe that there are three non-linear costs to the algorithm: managing votes, enumerating added sets, and executing the greedy algorithm on any remaining elements. We start by analyzing the cost of managing votes.

\begin{claim}\label{lem:boundOnActiveVoters}
There are at most $\OPT\cdot 100(\log(m)+\log(n))$ elements in $V$ after the main while loop terminates.
%The number of elements in $V$ after the main while loop terminates is at most $\OPT\cdot 100(\log(m)+\log(n))$
\end{claim}
\begin{proof}
    We know from \Cref{genAlgorithmInvariants} that $V\subseteq U_{\overline{C}}$.
    As in the proof of \Cref{lem:numberOfCallsToAddSet}, there must exist some set in $S\in \mathcal{F}\setminus C$ that covers at least a $1/\OPT$-fraction of the elements in $V$.
    Thus, if $|V| > \OPT\cdot 100(\log(m)+\log(n))$, this would imply that $v(S)\geq 100(\log(m)+\log(n))$ when the while loop terminated, contradicting the algorithm description.
\end{proof}

\begin{claim}\label{lem:costVotes}
    With probability at least $1-\frac{1}{m^{28}\cdot{n^{28}}}$, the cost of managing votes is at most 
    \[
    O\left(f \cdot (\log(m)+\log(n))\cdot \OPT\cdot\log(n)\right).
    \]
\end{claim}

\begin{proof}
    An element is added to $V$ at most once, and either is present at the end of the loop or removed during a call to \AddSet.
    By \Cref{lem:boundOnActiveVoters} at most $100\cdot \OPT\cdot (\log(m)+\log(n))$ elements are in $V$ at the end of the loop.
    Each call to $\AddSet(S)$ removes exactly $100(\log(m)+\log(n))$ elements in $V$, and by \Cref{lem:numberOfCallsToAddSet}, with probability at least $1-\frac{1}{m^{28}\cdot n^{28}}$ there are at most $O\left(\OPT\cdot \log(n)\right)$ calls to $\AddSet(S)$. 
    So, the total number of elements added to $V$ is $O((\log(m)+\log(n))\cdot \OPT\cdot\log(n))$.
    Each element added/removed from $V$ votes/unvotes for at most $f$ sets, which yields the result.
\end{proof}

\noindent We now bound the total cost of the algorithm.

\begin{lemma}\label{lem:RunningTimeBound}
    With probability at least $1-\frac{1}{m^{28}\cdot n^{28}}$, the total running time achieved by the algorithm is: 
    \[
    O(m + n + \OPT \cdot \log(n) \cdot (f \cdot (\log(m) + \log(n)) + \Delta)),
    \]
    where $f$ is the maximum frequency of an element in $\mathcal{S}_C$ and $\Delta$ is the maximum size of a set in $\mathcal{S}_C$:
\end{lemma}
\begin{proof}
    Observe that aside from costs linear in $n$ or $m$, the running time of the algorithm can broken into the cost of managing votes, the cost of enumerating added sets and the cost of the greedy algorithm at the end.
    By \Cref{lem:costVotes}, the cost of managing votes is  $O\left(f \cdot (\log(m)+\log(n))\cdot \OPT\cdot\log(n)\right)$.
    By \Cref{lem:numberOfCallsToAddSet}, the algorithm adds at most $O(\OPT\cdot \log(n))$ sets, each of which has at most $\Delta$ elements, so the cost of enumerating added sets is $O\left(\Delta \cdot \OPT\cdot \log(n)\right)$.
    By \Cref{lem:boundOnActiveVoters},  $|V|\leq O(\OPT\cdot (\log(m)+\log(n))$ and thus we can say $\sum_{S\in \mathcal{F\setminus C}}|S_{\overline{C}}|\leq O(f\cdot \OPT\cdot (\log(m)+\log(n))$ which produces a bound on the total time taken by the greedy algorithm.
\end{proof}

\subsection{Reducing cost by limiting the maximum frequency}
The cost of the voting algorithm depends both on $f$, the maximum frequency of an element in $\mathcal{S}_C$, and $\Delta$, the maximum size of a set in $\mathcal{S}_C$.
In this section, we'll give an algorithm that limits the effective value of $f$, which in turns reduces the cost.
%\subsection[First efficient algorithm]{An $O\left(m\cdot \log(n)\cdot  (\log(m)+\log(n))^2+n\cdot \OPT\cdot \log^2(n)\right)$-time algorithm}
%\label{firstGenericAlgorithm}
%This algorithm optimizes the running time of the prototypical algorithm by limiting the frequency of elements.
The key idea is to add $O(\OPT\cdot\log(n))$ random sets to the cover.
Then, any element that is a member of at least $\frac{m}{\OPT}$ sets will be covered with high probability, which in turn bounds $f$ by $\frac{m}{\OPT}$.
The issue is we don't know what $\OPT$ is equal to in advance, so we don't have a predetermined frequency threshold to enforce.
To overcome this we use exponential search.
We first make a guess $L$ for \OPT, allowing us to remove all sets that have at least $\frac{m}{L}$ elements.
We then run \GenVoteCover, which is cheaper, because $f$ is now small.
If our guess is sufficiently larger than \OPT, we show this will return a small set cover.
To deal with the case of our guess being smaller than \OPT, we modify \GenVoteCover to stop early, so that we don't incur its full running time.
In particular, we stop execution as soon as we add the $(L+1)$th set to the cover; we will refer to this as $\EarlyTermGVC(U,\mathcal{F}, L)$.
The algorithm is given as \Cref{alg:constructLimitedVoteCover}.

\begin{algorithm}
\caption{LimitedVoteCover($U,\mathcal{F}$)}
\label{alg:constructLimitedVoteCover}
\begin{algorithmic}[1]
\State $L\gets 16n\ln(n)$
\State $C^*\gets \varnothing$
\While{$L\geq 1$}
\ForEach{$e\in U$}
\If{$\FreqOf(e)\geq \frac{100\cdot(\log(m)+\log(n))\cdot m}{L}$}
\State $U \gets U\setminus{e}$
\EndIf
\EndFor
\State $C \gets \EarlyTermGVC(U,\mathcal{F},L)$.
\If{$C\neq$\FAIL}
\State $C^*\gets C$
\Else
\State Break while loop.
\State \Return{C}
\EndIf
\State $L\gets \frac{1}{2}\cdot L$
\EndWhile
\State Add $L$ random sets from $\mathcal{F}$ to $C^*$.
\State \Return{$C^*$}
\end{algorithmic}
\end{algorithm}

\paragraph{Correctness and approximation ratio.} We first show that \EarlyTermGVC is well-behaved.

\begin{lemma}
\label{lem:prototypicalCorrectness}
    \EarlyTermGVC either returns \FAIL or a set cover for $(U,\mathcal{F})$. Further, if $L> 9\cdot OPT\cdot \ln(n)$, then with probability at least $1-\frac{1}{m^{28}\cdot n^{28}}$ the algorithm produces a set cover with size $O(\OPT\cdot \ln(n))$, and if $L < \frac{1}{2}\cdot \OPT$ then the algorithm must return \FAIL.
\end{lemma}
\begin{proof}
    First, if we assume the algorithm does not return \FAIL, \EarlyTermGVC emulates \GenVoteCover exactly, and thus produces a set cover by the correctness established for \GenVoteCover. Now, the algorithm stops early when it adds the $(L+1)$th element to its cover.
    Therefore, if $L \geq 9\cdot \OPT\cdot \ln(n)$, by \Cref{lem:cvcApproximationRatio}, the algorithm returns a $9\ln(n)$-approximate set cover with probability at least $1-\frac{1}{m^{28}\cdot n^{28}}$.
    If $L<\frac{1}{2}\cdot \OPT$ then if the algorithm does not return \FAIL, it returns a set cover of size at most $2\cdot L < \OPT$, a contradiction.
\end{proof}

\noindent This in turn shows that \LimVoteCover returns a $O(\log(n))$-approximate set cover with high probability.

\begin{lemma}\label{lem:firstAlgorithmCorrectness}
With probability at least $1-\frac{1}{m^{27}n^{27}}$, \LimVoteCover produces a set cover for $(U,\mathcal{F})$ of size $O(\OPT\cdot\ln(n))$.
\end{lemma}

\begin{proof}
    By \Cref{lem:prototypicalCorrectness}, we can union bound over all calls to \EarlyTermGVC in which $L>8\ln(n)\cdot \OPT$ to conclude they all succeed and produce an $O(\ln(n))$-approximation with probability at least $1-\frac{n}{m^{28}n^{28}}\leq 1-\frac{1}{m^{28}n^{27}}$.
    Now, under the assumption that all these calls succeed, consider the first call to \EarlyTermGVC that fails with $L\leq 8\ln(n)\cdot \OPT$.
    Then the previous call to \EarlyTermGVC succeeded, and produced a set cover for $(U^\prime,\mathcal{F})$ where $U^\prime$ is the set of all elements with frequency at most $\frac{100\cdot  (\log(m)+\log(n))\cdot m}{2L}$ in $U$.
    Further, note that because \EarlyTermGVC stops early, the size of the returned cover, $C^*$, is at most $L\in O(\ln(n)\cdot \OPT)$.  \LimVoteCover then adds $L$ random sets to $C^*$.
    The probability for a given element in $U\setminus U^\prime$ to remain uncovered by $C^*$ following this addition is at most $\left(1-\frac{100\cdot (\log(n)+\log(m))}{2L}\right)^{L}\leq \frac{1}{n^{50}m^{50}}$.
    So, we can union bound over all elements in $U\setminus U^\prime$ to conclude the probability any such element remains uncovered is at most $\frac{1}{n^{49}m^{50}}$
    Thus by applying union bound the probability the algorithm succeeds with an $O(\OPT\cdot\ln(n))$ approximation ratio is at least $1-\frac{1}{m^{28}n^{27}}-\frac{1}{n^{49}m^{50}}\geq 1-\frac{1}{m^{27}n^{27}}$
\end{proof}

\paragraph{Running time analysis.}
We know from \cref{lem:prototypicalCorrectness} that when $L<\frac{1}{2}\cdot \OPT$ the call to \EarlyTermGVC fails and we return. 
Thus the running time is dominated by the calls for $L\geq \OPT$.
We can union bound over all of these calls, to conclude that with probability at least $1-\frac{1}{m^{28}n^{27}}$, \Cref{lem:RunningTimeBound} holds for all of them, thus bounding the running time for each individual call.
We can therefore conclude a total bound of:
\[
O\left(\sum_{k=\log(\OPT)}^{\log(16n\cdot \ln(n))} m+n+\OPT\cdot \log(n)\cdot (f_k\cdot (\log(m)+\log(n))+\Delta)\right),
\]
where $f_k$ is the maximum frequency in the call to \EarlyTermGVC when $L=2^k$.
By the criterion the algorithm imposes we have:
\[
f_k\leq \frac{100(\log(m)+\log(n))\cdot m}{L}\leq %\frac{100(\log(m)+\log(n))\cdot m}{2^{\log(\OPT)}}=
\frac{100(\log(m)+\log(n))\cdot m}{\OPT}.
\]
Therefore with probability at least $1-\frac{1}{m^{28}n^{27}}$ the running time is upper bounded by:
\begin{align*}
O\left(\sum_{k=\log(\OPT)}^{\log(16n\ln n)} m\log n(\log m+\log n)^2 + \OPT\cdot n\log n\right)
= O\left(m\log^2 n(\log m+\log n)^2 + \OPT\cdot n\log^2 n\right)
\end{align*}

\noindent We have shown:
\begin{lemma}
\label{lem:limitedCost}
    With probability at least $1-\frac{1}{m^{28}n^{27}}$, \LimVoteCover returns an $O(\log n)$ approximate set cover in $O\left(m\log^2 n(\log m+\log n)^2 + \OPT\cdot n\log^2 n\right) = \tilde{O}(m + \OPT \cdot n)$ time.
\end{lemma}

\subsection{Using lazy enumeration to save cost}
We now turn to another way to reduce the cost of the voting algorithm, namely avoiding the expensive enumeration during the \AddSet step.
The idea is that rather than marking elements as covered in the \EarlyTermGVC algorithm, we can instead check whether an element has already been covered during the sampling step (when drawing it from $N$) and remove it then.
This can be achieved using $\SetOf$ queries, and testing whether each set the element is a member of is in $C$ or not.
We refer to this algorithm as \LazyETGVC and the version of \LimVoteCover which uses it as \LazyLimVoteCover.
\LazyETGVC is logically equivalent to \EarlyTermGVC: the only difference is when elements are identified as covered.
Therefore the correctness and approximation ratio properties of \Cref{lem:prototypicalCorrectness} are still satisfied.
Similarly, \LazyLimVoteCover satisfies \Cref{lem:firstAlgorithmCorrectness}.

\begin{lemma}\label{lem:RunningTimeBoundLazyCVCSE}
    The total running time for \LazyETGVC is $O(m+n\cdot f)$ where $f$ is the maximum frequency of an element in $\mathcal{S}_C$.
\end{lemma}

\begin{proof}
    In contrast to the proof of \Cref{lem:RunningTimeBound}, the non-linear cost is dominated by the cost of sampling voters (now requiring checking whether they are covered), vote management, and the greedy algorithm at the end.
    Note that the cost of \AddSet is now constant, since we are no longer enumerating the set. 
    The cost of checking whether a potential voter has already been covered is at most $f$, since we check whether any of the at most $f$ sets it belongs to have been chosen.  The cost of vote management is also at most $f$.
    So we can bound both of these costs in total by $O(n\cdot f)$.
    The cost of executing the greedy algorithm is $\sum_{S\in \mathcal{F\setminus C}}|S_{\overline{C}}|\leq  \sum_{e\in U }|\{S\in\mathcal{F}:e\in S\}|\leq  O(n\cdot f)$, which completes the proof.
\end{proof}
We now prove the following bound on the running time of \LazyLimVoteCover:
\begin{lemma}
\label{lem:lazyRunningTime}
    The running time of \LazyLimVoteCover is bounded by $O\left(\frac{mn\cdot\log(n)\cdot (\log(m)+\log(n))}{\OPT}\right) = \tilde{O}(mn/\OPT)$.
\end{lemma}

\begin{proof}
As in the analysis of \LimVoteCover, the running time is dominated by calls to \LazyETGVC in which $L\geq \OPT$.
We can therefore conclude according to \cref{lem:RunningTimeBoundLazyCVCSE} that the running time is bounded by 
\[
O\left(\sum_{k=\log(\OPT)}^{\log(16n\cdot \ln(n))} m+n\cdot f_k \right),
\]
where $f_k$ is the maximum frequency when $L=2^k$.
As before,
\[
f_k\leq \frac{100\cdot(\log(m)+\log(n))\cdot m}{L}\leq \frac{100\cdot(\log(m)+\log(n))\cdot m}{\OPT},
\]
so the running time is upper bounded by:
\begin{align*}
O\left(\sum_{k=\log(\OPT)}^{\log(16n\cdot \ln(n))} m+n\cdot{f_k} \right)&\leq O\left(\sum_{k=\log(\OPT)}^{\log(16n\cdot \ln(n))} m+\frac{mn\cdot (\log(m)+\log(n))}{\OPT} \right) \\
&\leq O\left(\frac{mn\cdot\log(n)\cdot (\log(m)+\log(n))}{\OPT}\right).\qedhere
\end{align*}
\end{proof}

\Cref{thm:general_sublinear} now follows as an immediate consequence of \Cref{lem:limitedCost,lem:lazyRunningTime}, by taking the minimum runtime. A simple theoretical approach to doing this it to run \LimVoteCover and \LazyLimVoteCover in parallel, then take the result of the first to finish. This gives us an algorithm that constructs an $O(\log n)$-approximate minimum set cover in time proportional to the minimum running time of the two algorithms. This gives us an $\tilde O(\min(m + n\OPT{}, mn/\OPT{}))$ time algorithm.